\documentclass[submission,copyright,creativecommons]{eptcs}
\providecommand{\event}{AiML 2026} 

\usepackage{aiml26}
\usepackage{iftex}

\ifpdf
  \usepackage{underscore}         % Only needed if you use pdflatex.
  \usepackage[T1]{fontenc}        % Recommended with pdflatex
\else
  \usepackage{breakurl}           % Not needed if you use pdflatex only.
\fi

\usepackage{tikz}
\usepackage{graphicx}
\usepackage{amssymb}
\usepackage{amsmath}
\usepackage{xcolor}
\usepackage{relsize}
\usepackage{mdframed}
\usepackage{linguex}
\tikzset{
  opaque/.style={
    fill=gray,
    fill opacity=.1
}}

\title{Inquisitive Action Logic
}
\author{Ivano Ciardelli
\institute{University of Padua, Italy}
\email{ivano.ciardelli@unipd.it}}

\newcommand{\titlerunning}{Inquisitive Action Logic}
\newcommand{\authorrunning}{I. Ciardelli}

\hypersetup{
  bookmarksnumbered,
  pdftitle    = {\titlerunning},
  pdfauthor   = {\authorrunning},
  pdfsubject  = {EPTCS},               % Consider adding a more appropriate subject or description
  pdfkeywords = {inquisitive modal logic, logic of action, determination, actual effectivity functions} % Uncomment and enter keywords specific to your paper
}

\newcommand{\inqal}{\textsf{InqAL}}
\newcommand{\inqnl}{\textsf{InqNL}}
\newcommand{\inqnla}{\ensuremath{\textsf{InqNL}_{\A}}}

\newcommand{\To}{\Rrightarrow}
\newcommand{\Toa}{\Rrightarrow_{a}\!}

\newcommand{\A}{\ensuremath{\mathcal{A}}}
\renewcommand{\L}{\ensuremath{\mathcal{L}}}
\newcommand{\K}{\ensuremath{\mathcal{K}}}
\newcommand{\La}{\ensuremath{\mathcal{L}_{\mathcal{A}}}}
\newcommand{\Lan}{\ensuremath{\mathcal{L}_{\mathcal{A}}^n}}
\newcommand{\Lad}{\ensuremath{\mathcal{L}_{\mathcal{A}}^!}}
\newcommand{\Land}{\ensuremath{\mathcal{L}_{\mathcal{A}}^{!n}}}
\renewcommand{\P}{\ensuremath{\mathcal{P}}}
\newcommand{\R}{\ensuremath{\mathcal{R}}}

\newcommand{\lori}{\mathbin{\rotatebox[origin=c]{-90}{$\geqslant$}}}
\def\DISJI{\ensuremath{\mathlarger{\mathbin{{\setminus}\mspace{-5mu}{\setminus}}\hspace{-0.33ex}/}}\xspace}
\newcommand{\Lori}{\DISJI}
 
\newcommand{\quotes}[1]{``#1''}
\renewcommand{\phi}{\varphi}
\newcommand{\md}{\text{md}}

\DeclareFontFamily{U}{MnSymbolC}{}
\DeclareSymbolFont{MnSyC}{U}{MnSymbolC}{m}{n}
\DeclareMathSymbol{\diamondplus}{\mathbin}{MnSyC}{"7C}
\DeclareFontShape{U}{MnSymbolC}{m}{n}{
    <-6>  MnSymbolC5
   <6-7>  MnSymbolC6
   <7-8>  MnSymbolC7
   <8-9>  MnSymbolC8
   <9-10> MnSymbolC9
  <10-12> MnSymbolC10
  <12->   MnSymbolC12}{}

\newcommand{\ibox}{\boxplus}
\newcommand{\wbox}{\boxtimes}
\newcommand{\idia}{\diamondplus}

\renewcommand{\aa}{\alpha}
\newcommand{\bb}{\beta}
\renewcommand{\gg}{\gamma}

\newcommand{\ActP}{\textsf{ActP}}
\newcommand{\act}{\textsf{act}}
\newcommand{\out}{\textsf{out}}

\let\Sec\S
\renewcommand{\S}{\ensuremath{\mathcal{S}}}

\newcommand{\F}{F}
\newcommand{\M}{\ensuremath{\mathcal{M}}}

\newcommand{\vdashn}{\ensuremath{\vdash_{\textsf{N}}}}
\newcommand{\vdasha}{\ensuremath{\vdash_{\textsf{A}}}}

\newcommand{\preqn}{\ensuremath{\dashv\vdash_{\textsf{N}}}}
\newcommand{\preqa}{\ensuremath{\dashv\vdash_{\textsf{A}}}}

\newcommand{\Togamma}{\Rrightarrow_\Gamma^a}

\newcommand{\Left}{\textsf{L}}
\newcommand{\Right}{\textsf{R}}

\begin{document}
\maketitle

\begin{abstract} We introduce \emph{inquisitive action logic}, \inqal, a multi-agent modal logic for reasoning about action. While traditional approaches focus on what \emph{properties} of the outcome an agent can \emph{force}, \inqal\ also captures what \emph{aspects} of the outcome an agent \emph{determines} through their actions. As we argue, such claims of agentive determination are naturally analyzed as modal claims involving questions. 

Technically, \inqal\ is a multi-agent extension of inquisitive neighborhood logic \cite{Ciardelli:25neighborhood} based on concurrent game structures. With respect to statements, it is expressively equivalent to the individual-agent fragment of the socially friendly coalition logic recently proposed by Goranko and Enqvist~\cite{GorankoEnqvist:18}. 

We present an axiomatization of \inqal\ and prove completeness and decidability via the finite model property. Along the way, we establish a representation theorem for \emph{actual effectivity functions}, associating to an agent the sets of outcomes corresponding to their possible actions; we give exact conditions under which a multi-agent neighborhood frame arises from a concurrent game structure.
\end{abstract}

%%%%%%%%%%%%%%%%%%%%%%%%%%%%%%%%%%%%%%%
%%%%%%%%%%%%%%%%%%%%%%%%%%%%%%%%%%%%%%%

\section{Introduction}
\label{sec:intro}

There is a long tradition of using modal logics to reason about actions, outcomes, and abilities (for surveys, see \cite{Agotnes:15,Segerberg:16}). This enterprise is motivated by philosophical, juridical, and computational concerns: logics of agency have been used, e.g., %to 
%sharpen debates about the metaphysics of time \citep[see, e.g.,][]{Correia:12}, 
to spell out precise notions of responsibility \cite{Lorini:14,Baltag:21}, and for the automated verification of software properties \cite{Demri:16}.
Logics in this tradition typically focus on the powers of agents (or coalitions) in an interactive situation, i.e., on the facts they are able to force through their actions. Something which is just as important, however, is what aspects of the world an agent \emph{determines}, or \emph{influences}, through their actions.
For an illustration, imagine that Alice and Bob are creating clay animals together; Alice moulds an animal with the clay, and Bob paints it. While Alice works, Bob picks a color, without looking at what Alice is moulding. In this scenario, it is natural to say the following:

\ex. 
\a. Alice determines the shape of the animal (but not its color).
\b. Bob determines the color of the animal (but not its shape).

These facts matter, among other things, for ascribing responsibility: suppose little Charlie asks his siblings to make a red cat for him; if what he gets is a blue cat, he should complain to Bob, not to Alice.

How can we analyze a statement like \Last[a]? More precisely, what is the \emph{thing} that Alice determines? Standardly, in modal logic, modalities apply to propositions. However, ``the shape of the animal'' does not denote a proposition; rather, it is associated with a set of propositions: that the animal is a dog, that the animal is a cat, etcetera. What Alice determines is not a specific one of these propositions, but \emph{which} among these propositions is going to be actualized. In other words, what Alice determines is the answer to a question, namely, the question ``what the shape is''. Determining, then, can be seen as a relation between an agent and a question. What does this relation amount to? Here is a natural idea: an agent can be said to determine a question if the answer to the question is not settled \emph{a priori}, but it becomes settled once we fix the agent's action. Implementing this idea takes us into the realm of \emph{inquisitive modal logic}, a quickly growing research area that combines inquisitive logic and modal logic.

Inquisitive logic \cite{Ciardelli:23book} is a framework that allows one to build logical systems including not only, as usual, formulas expressing propositions (e.g., \emph{that the animal is a cat}), but also formulas expressing questions (e.g., \emph{what the shape of the animal is}). %or \quotes{what the color is}). 
One can then define inquisitive modalities that apply to such formulas (see, a.o., \cite{Ciardelli:14aiml,CiardelliRoelofsen:15idel,Ciardelli:18aiml,Gessel:20action,Gessel:21,PuncocharSedlar:21epistemic,PuncocharSedlar:21pdl,CiardelliOtto:21,MeissnerOtto:22,MaricPerkov:24,Ciardelli:25neighborhood}). 
In our case, the idea is to capture agentive determination in terms of an inquisitive modality $\wbox$ indexed by agents %(e.g., $a$ for Alice and $b$ for Bob) 
which, when applied to formulas \textsf{shape} and \textsf{color} formalizing the questions \emph{what is the shape} and \emph{what the color}, delivers modal statements $\wbox_{\text{Alice}}\textsf{shape}$ and $\wbox_{\text{Bob}}\textsf{color}$ capturing, respectively, the claims in \Last[a] and \Last[b].

In this paper, we will develop this idea in detail. We will define an \emph{inquisitive action logic}, \inqal, which allows us to analyze statements like \Last[a] and \Last[b] along the lines we described, but which also recovers the modal account of ability familiar from the classical work of Brown \cite{Brown:88} and more recent work on coalition logic, \textsf{CL} \cite{Pauly:02}. Our logic is a multi-agent version of a recently developed system, \emph{inquisitive neighborhood logic}, \inqnl\ \cite{Ciardelli:22aiml,Ciardelli:25neighborhood}. \inqnl\ is based on a binary modality $\To$, acting semantically as a strict conditional quantifying over neighborhoods; analogously, our multi-agent extension \inqal\ contains one such modality $\Toa$ for each agent $a$; in terms of this modality, some derivative unary modalities (including the modality $\wbox_a$ mentioned above) can be defined. Besides being multi-agent, \inqal\ is geared specifically towards an interpretation in terms of action; as a consequence, while \inqnl\ is interpreted over arbitrary neighborhood models, \inqal\ starts with more concrete structures known as \emph{concurrent game structures} or \emph{multi-agent transition systems} \cite{Alur:02,Pauly:02,GorankoDrimmelen:06,Bulling:16}, modeling processes in which a system transitions from a state to another based on the joint actions of multiple agents. Now, from such models one can extract a corresponding neighborhood model, where the neighborhoods for a given agent correspond to the actions available to them, each neighborhood encoding the range of outcomes that may possibly result if the agent performs that action. The neighborhood maps induced in this way are known in the literature as \emph{actual effectivity functions} \cite{Bulling:16} (in contrast with the \emph{effectivity functions} studied in coalition logic \cite{Pauly:02,Goranko:13}, which are upward-closed sets of neighborhoods). The problem of characterizing which families of neighborhood maps are the actual effectivity functions of some concurrent game structure is currently open. In this paper, we settle this problem for the case in which we only consider single agents, and not proper coalitions. This result is interesting in its own right, but it also plays a key role in our study of \inqal, giving us a characterization of the relevant class of neighborhood models.

Our logic \inqal\ is also connected to \emph{socially friendly coalition logic} (\textsf{SFCL}), an extension of coalition logic recently proposed by Goranko and Enqvist \cite{GorankoEnqvist:18}. Whereas standard coalition logic focuses on what outcomes agents can \emph{force} through their actions, \textsf{SFCL} also captures what outcomes agents can \emph{enable} for others, while still forcing their own goals. These properties are crucial for the possibility of cooperation between agents (hence the name \emph{socially friendly}). Technically, \textsf{SFCL} is a multi-agent extension of \emph{instantial neighborhood logic}, \textsf{INL} \cite{Benthem:17}. As proved in \cite{Ciardelli:25neighborhood}, inquisitive neighborhood logic has, with respect to statements, the same expressive power as \textsf{INL}. In the same vein, our inquisitive action logic has, with respect to statements, the same expressive power as \textsf{SFCL}. Still, as we will see, the way modal claims are expressed in these logics is very different; in particular, the translation from \inqal\ to \textsf{SFCL} can lead to an exponential increase in the size of the formula.

The paper is structured as follows. In \Sec2 we provide a characterization of those neighborhood frames that arise as actual effectivity functions from some concurrent game structure. In \Sec3 we introduce \inqnla, a multi-agent version of inquisitive neighborhood logic. In \Sec4 we introduce \inqal, which interprets \inqnla\ over concurrent game models, and discuss how this logic allows us to express interesting properties relating to what agents can force, enable, or determine in an interactive situation. In \Sec5 we give an axiomatization of \inqal\ and prove completeness and decidability. In \Sec6 we compare the expressive power of our logic with that of Socially Friendly Coalition Logic. \Sec7 concludes.

%%%%%%%%%%%%%%%%%%%%%%%%%%%%%%%%%%%%%%%
%%%%%%%%%%%%%%%%%%%%%%%%%%%%%%%%%%%%%%%

\section{Characterizing actual effectivity functions of agents}

Concurrent game structures \cite{Alur:02,Pauly:02,Bulling:16} model processes in which the transition from one stage to the next is determined by the simultaneous actions of multiple agents. The formal definition is as follows. %This idea is captured formally by the following definition. %The formal definition is as follows. 

\begin{definition} Let $\A=\{a_1,\dots,a_k\}$ be a finite set of agents.  A \emph{concurrent game structure} (or CGS for short) for $\A$ is a triple $\S=(W,\act,\out)$ where:
\begin{itemize}
%\item $\Agt=\{1,\dots,k\}$ for some $k$ is a set of agents
\item $W$ is a non-empty set of \emph{worlds}, representing stages of the process;
%and $\Act$ a non-empty set of \emph{actions}.
\item $\act$ is a function which assigns to each agent $a$ and world $w$ a non-empty set $\act(a,w)$ of actions.
%:\A\times W\to\wp(\Act)-\{\emptyset\}$ assigns to each agent and world a nonempty set of available actions. 

A function $\tau$ assigning to each agent $a_i$ a corresponding action $\tau(a_i)\in\act(a_i,w)$ is called an \emph{action profile} at $w$; the set of action profiles at $w$ is denoted $\ActP(w)$.
\item $\out$ is a function that, given a world $w$ and an action profile $\tau$ at $w$, returns a world $\out(w,\tau)$; intuitively, $\out(w,\tau)$ is the world that results from $w$ when each agent $a_i$ performs the action $\tau(a_i)$.
%\item $V:\P\to\wp(W)$ is a valuation function, assigning to each atom a set of worlds where it is true.
\end{itemize}
\end{definition}
\noindent
At a world $w$ in a CGS, a range of outcomes are possible, depending on the choices of the agents, namely: $$O(w)=\{\out(w,\tau)\mid\tau\in\ActP(w)\}$$ 
By picking a particular action, an agent typically narrows down the set of possible outcomes. The set of outcomes which are possible from $w$ given that agent $a$ performs action $\tau_a\in\act(a,w)$ is:
$$O_{a}(w,\tau_a)=\{\out(w,\tau)\mid\tau\in\ActP,\tau(a)=\tau_a\}$$
We call this the \emph{output set} of action $\tau_a$ for $a$ at $w$. 
By collecting the output sets of all actions available to an agent, we obtain their \emph{actual effectivity function}. More precisely, the actual effectivity function of agent $a$ in a CGS $\S$ is the map $\Sigma_a^\S:W\to\wp\wp(W)$ given by:
$$\Sigma_a^\S(w)=\{O_a(w,\tau_a)\mid \tau_a\in\act(a,w)\}$$
In this way, a CGS determines a corresponding multi-agent neighborhood frame $\F_\S=(W,(\Sigma_a^\S)_{a\in\A})$. But not just \emph{any} neighborhood frame $F=(W,(\Sigma_a)_{a\in\A})$, where $\Sigma_a:W\to\wp\wp(W)$, can be a system of actual effectivity functions of some CGS. This raises a natural question: \emph{which} neighborhood frames arise from some CGS? The following theorem provides an answer to this question.\footnote{A reviewer asks why it is difficult to extend this result from individual agents to coalitions. The reason is that coalitions can be related to each other by inclusion or overlap; in these cases, there are complex interplays between their effectivity functions.}

\begin{theorem}\label{theor:characterization} Let $\A=\{a_1,\dots,a_k\}$ be a finite set of agents with $k>1$.\footnote{\label{fn:single agent characterization}When $k=1$, i.e., in the single-agent case, the characterization is trivial: as the reader can readily verify, a neighborhood frame $\F=(W,\Sigma)$ is induced by a CGS if and only if $\Sigma(w)$ is a non-empty set of singletons for every $w\in W$. We therefore focus on the interesting multi-agent case with $k>1$.} 
For a multi-agent neighborhood frame $\F=(W,(\Sigma_a)_{a\in\A})$ the following are equivalent:
\begin{enumerate}
\item $\F=\F_\S$ for some concurrent game structure $\S$.
\item The following three conditions are satisfied for each world $w\in W$:
\begin{enumerate}
\item Existence of actions: for all $a\in\A$, $\Sigma_a(w)\neq\emptyset$;
%\item Non-empty range: for all $a\in\Agt$, $\bigcup\Sigma_a(w)\neq\emptyset$
\item Uniform range: for all $a,b\in\A$, $\bigcup\Sigma_a(w)=\bigcup\Sigma_b(w)$;
\item Independence: for all $s_1\in\Sigma_{a_1}(w),\dots,s_k\in\Sigma_{a_k}(w): s_1\cap\dots\cap s_k\neq\emptyset$.
\end{enumerate}
\end{enumerate}
\end{theorem}

\begin{proof} $1\Rightarrow 2$ amounts to the claim that the actual effectivity functions of a CGS satisfy conditions (a)-(c) in 2. Condition (a) follows from the fact that $\act(a,w)$ is non-empty by definition. Condition (b) is due to the observation that for every $a$, $\bigcup\Sigma_a^\S(w)$ coincides with the set of all possible outcomes at $w$, $O(w)$. Condition (c) is due to the fact that whenever we take neighborhoods $O_{a_i}(w,\tau_{a_i})\in\Sigma_{a_i}^\S(w)$ for $i=1,\dots,k$, the function $\tau$ defined by $\tau(a_i)=\tau_{a_i}$ is an action profile, and we have $\out(w,\tau)\in O_{a_i}(w,\tau_{i})$ for each $i$. 

To show $2\Rightarrow 1$, suppose $\F=(W,(\Sigma_a)_{a\in\A})$ is a neighborhood frame satisfying (a)-(c). Our aim is to define a CGS $\S=(W,\act,\out)$ with $\F_\S=\F$, i.e., such that for each agent $a$, $\Sigma_a$ coincides with the actual effectivity function $\Sigma_a^\S$ of agent $a$ in $\S$. 

For this, equip the set $W$ with a group structure, so that $(W,e,*,(\cdot)^{-1})$ is a group. Also, fix a well-ordering of $W$, and for any non-empty $s\subseteq W$, let $\min(s)$ denote the least element of $s$ according to this well-ordering.\footnote{Given the axiom of choice, it is well-known that any set admits both a well-ordering and a group structure. No appeal to the axiom of choice is needed if $W$ is finite or countably infinite, as will be the case in our completeness proof in Section 5.} We define the actions available to an agent $a$ at a world $w$ as pairs consisting of a neighborhood for $a$, and a world:
$$\act(a,w)=\{(s,v)\mid s\in\Sigma_a(w), v\in W\}$$
Condition (a) guarantees that this is a non-empty set, as required by the definition of a CGS.

Now given a world $w$, an action profile at $w$ can be identified with a sequence $\tau=((s_1,v_1),\dots,(s_k,v_k))$ where $s_i\in\Sigma_{a_i}(w)$ for $i\le k$. We define the outcome of $\tau$ at $w$ in the following way:
$$\out(w,\tau)=\left\{
\begin{array}{ll}
v_1*\dots*v_k &\text{if }(v_1*\dots*v_k)\in (s_1\cap\dots\cap s_k)\\
\min(s_1\cap\dots\cap s_k) & \text{otherwise}
\end{array}
\right.$$
Condition (c) guarantees that $s_1\cap\dots\cap s_k$ is non-empty, so that $\min(s_1\cap\dots\cap s_k)$ exists, which ensures that $\out(w,\tau)$ is well-defined. Note further that, by definition, $\out(w,\tau)\in s_1\cap\dots \cap s_k$.

We claim that the CGS so defined induces $\F$. To show this, the key step is to prove that for each agent $a$ and action $(s,v)\in\act(a,w)$, the set of outputs that may result from $a$ performing $(s,v)$ is simply~$s$:
\begin{equation}\label{eq:key step}
O_a(w,(s,v))=s
\end{equation}
The inclusion $\subseteq$ is clear: whenever $\tau$ is an action profile that includes the action $(s,v)$, the corresponding outcome $\out(w,\tau)$ is bound to be in $s$ by definition of $\out$. To show the inclusion $\supseteq$, consider an arbitrary world $u\in s$. We need to show that there is an action profile $\tau$ including $(s,v)$ whose outcome is $u$. 

For simplicity, suppose $a=a_1$ and let us write $(s_1,v_1)$ for $(s,v)$. Since $u\in s_1$ and $s_1\in\Sigma_{a_1}(w)$, we have $u\in\bigcup\Sigma_{a_1}(w)$. By condition (b), for each $i>1$ we have $u\in\bigcup\Sigma_{a_i}(w)$, so we can pick a neighborhood $s_i\in\Sigma_{a_i}(w)$ with $u\in s_i$. Now for $1<i<k$, pick some worlds $v_i\in W$ arbitrarily. Finally, let $v_k$ be the element $(v_1*\dots*v_{k-1})^{-1}*u$. Now consider the action profile $\tau=((s_1,v_1),\dots,(s_k,v_k))$.
We have $v_1*\dots*v_k=(v_1*\dots*v_{k-1})*(v_1*\dots*v_{k-1})^{-1}*u=u$. By the way the sets $s_i$ were chosen we know that $u\in s_1\cap\dots\cap s_k$, and so by definition of $\out$ we have $\out(w,\tau)=u$. Since the profile $\tau$ includes the action $(s,v)$, this shows that $u\in O_a(w,(s,v))$, as required to complete the proof of (\ref{eq:key step}).

Finally, using (\ref{eq:key step}), the actual effectivity function for agent $a$ turns out to be
$$\Sigma_a^\S(w)=
\{O_a(w,\tau_a)\mid\tau_a\in\act(a,w)\}=
\{O_a(w,(s,v))\mid s\in \Sigma_a(w),v\in W\}=
\{s\mid s\in \Sigma_a(w)\}=
\Sigma_a(w)$$
which shows that indeed, $\F_\S=\F$, as desired.
\end{proof}

%%%%%%%%%%%%%%%%%%%%%%%%%%%%%%%%%%%%%%%
%%%%%%%%%%%%%%%%%%%%%%%%%%%%%%%%%%%%%%%

\section{Inquisitive multi-agent neighborhood logic}
\label{sec:inqnl}

In this section we present a multi-agent version of inquisitive neighborhood logic, \inqnl\ \cite{Ciardelli:25neighborhood}. We call this logic $\inqnla$ where $\A=\{a_1,\dots,a_k\}$ is a finite set of agents. The proofs of the results in this section are the same as for \inqnl; we do not repeat them here, but include pointers to the results in \cite{Ciardelli:25neighborhood}. 

\medskip\noindent\textbf{Syntax.} Given a finite set $\P$ of atoms, the set $\La$ of \inqnla-formulas is given~by:
$$\phi\;::=\;p\mid\bot\mid(\phi\land\phi)\mid(\phi\to\phi)\mid(\phi\lori\phi)\mid(\phi\Toa\phi)$$
where $p\in\P$ and $a\in\A$. As standard in inquisitive logic, we also use the following abbreviations: $\neg\phi:=(\phi\to\bot)$; $\top:=\neg\bot$; $(\phi\lor\psi):=\neg(\neg\phi\land\neg\psi)$; ${?\phi}:=(\phi\lori\neg\phi)$.
The connective $\lori$, called \emph{inquisitive disjunction}, is regarded as a question-forming disjunction. Thus, e.g., the formula ${?p}$ (short for $p\lori\neg p$) is regarded as the question \emph{whether or not $p$}. The operators $\To_a$ are indexed versions of the binary modality $\To$ \inqnl. Two unary modalities (read \quotes{window} and \quotes{kite}) are defined in terms of $\Toa$ as follows:
$$\ibox_a\phi:=(\top\Toa\phi)\qquad \idia_a\phi:=\neg(\phi\Toa\bot)$$
An important fragment of the language is given by \emph{declarative formulas} (or \emph{declaratives}, for short), where inquisitive disjunction occurs only in the scope of a modal operator. More formally, the set \Lad\ of declaratives is defined as follows, where $\phi$ can be rewritten as any formula from $\La$:
$$\aa\;::=\;p\mid\bot\mid(\aa\land\aa)\mid(\aa\to\aa)\mid(\phi\Toa\phi)$$
Intuitively, declaratives stand for statements (as we will see, this is reflected in a key semantic property). We will use the meta-variables $\phi,\psi,\chi$ for arbitrary formulas, and $\aa,\bb,\gg$ for declaratives.

The \emph{modal depth} of a formula is defined as usual as the maximum number of nestings of modal operators in it. More precisely, we set: $\md(p)=\md(\bot)=0$; $\md(\phi\land\psi)=\md(\phi\to\psi)=\md(\phi\lori\psi)=\max(\md(\phi),\md(\psi))$;  $\md(\phi\Toa\psi)=\max(\md(\phi),\md(\psi))+1$. The set of formulas in $\La$ with modal depth up to $n$ is denoted $\Lan$; the set of declaratives with modal depth up to $n$ is denoted $\Land$.

\medskip\noindent\textbf{Semantics.} We interpret formulas relative to \emph{multi-agent inhabited neighborhood models} (main-models, for short), which are triples $M=(W,(\Sigma_a)_{a\in\A},V)$ where $W\neq\emptyset$ is a set of worlds, $\Sigma_a:W\to\wp\wp_{0}(W)$ is a map assigning to each $w\in W$ a set $\Sigma_a(w)$ of non-empty subsets of $W$ (called the \emph{neighborhoods} of $w$) and $V:\P\to\wp(W)$ is a valuation function. 
As customary in inquisitive logic, the semantics of \inqnla\ is not given in terms of truth relative to possible worlds, but instead in terms of \emph{support} relative to sets of possible worlds, called \emph{information states} (or simply \emph{states}).

\begin{definition}[Support in main-models] Let $M=(W,(\Sigma_a)_{a\in\A},V)$ be a main-model. The relation of support between information states $s\subseteq W$ and formulas $\phi\in\La$ is defined as follows:
\begin{itemize}
\item $M,s\models p\iff s\subseteq V(p)$
\item $M,s\models\bot\iff s=\emptyset$
\item $M,s\models\phi\land\psi\iff M,s\models\phi$ and $M,s\models\psi$
\item $M,s\models\phi\lori\psi\iff M,s\models\phi$ or $M,s\models\psi$
\item $M,s\models\phi\to\psi\iff \forall t\subseteq s: M,t\models\phi$ implies $M,t\models\psi$
\item $M,s\models\phi\Toa\psi\iff \forall w\in s\forall t\in\Sigma_a(w):M,t\models\phi$ implies $M,t\models\psi$
\end{itemize}
\end{definition}

\noindent
All the clauses except for the last one are standard in inquisitive logic (see \cite{Ciardelli:23book} for discussion), while the clause for the modality $\Toa$ is a direct multi-agent adaptation of the one for $\To$ in \inqnl\ \cite{Ciardelli:25neighborhood}. %it says that for all worlds in the evaluation state, all neighborhoods that support the antecedent also support the consequent. %The only difference is that, whereas in \inqnl\ the neighborhoods are given directly by the model, here they are determined by the game structure $\S$ as different output sets corresponding to the different actions available to $a$ in $w$. In other words, the above semantics interprets the modality $\Toa$ via the particular neighborhood model $\M=(\F_\S,V)$ induced by our concurrent game model.

Entailment and equivalence are defined in the obvious way: a set of formulas $\Phi$ entails a formula $\psi$, denoted $\Phi\models_{\inqnla}\psi$, if for every main-model $M$ and state $s$ that supports all formulas in $\Phi$, $s$ also supports $\psi$. Two formulas $\phi,\psi$ are equivalent, denoted $\phi\equiv_{\inqnla}\psi$, if they are supported by the same states in every main-model. To ease notation, throughout this section we drop the subscript \inqnla.

As usual in inquisitive logic, support is persistent (if $M,s\models\phi$ and $t\subseteq s$, then $M,t\models\phi$), and the empty state trivially supports every formula. Moreover, we retrieve a notion of truth at a world $w$ by defining it as support relative to the corresponding singleton $\{w\}$.

\begin{definition}[Truth] $\phi$ is \emph{true} at a world $w$ of a main-model $M$, denoted $M,w\models\phi$, in case $M,\{w\}\models\phi$. The set of worlds at which $\phi$ is true in a model $M$ is denoted $|\phi|_M$.
\end{definition}

\noindent
Some simple calculations show that the Boolean connectives have their usual truth-functional behavior (for instance, $M,w\models\neg\phi\iff M,w\not\models\phi$), and that modal formulas have the following truth conditions:%\footnote{Note that the clause for $\idia$ makes crucial use of the fact that neighborhoods are assumed to be non-empty.}
\begin{itemize}
\item $M,w\models(\phi\Toa\psi)\iff \forall s\in\Sigma_a(w):M,s\models\phi$ implies $M,s\models\psi$
\item $M,w\models\ibox\phi\iff \forall s\in\Sigma_a(w):M,s\models\phi$
\item $M,w\models\idia\phi\iff \exists s\in\Sigma_a(w):M,s\models\phi$
\end{itemize}

%$M,w\models{\phi\land\psi}\iff M,w\models\phi$ and $M,w\models\psi$. 
%The supportSupport conditions determine truth conditions, but the converse is not the case in general.  not the other way around. 
\noindent
Declaratives have a special semantic property: for them, support boils down to truth at each world. %Thus, the semantics of a declarative is fully determined by its truth conditions. 
In fact, up to equivalence, declaratives are \emph{exactly} the formulas for which this holds. 
%The inductive proof of this fact is standard and, thus, omitted (see \cite{Ciardelli:25neighborhood} for the details in the case of \inqnl).

\begin{definition}[Truth-conditionality] A formula $\phi\in\La$ is \emph{truth-conditional} if for every model $M$ and state $s$ we have $M,s\models\phi\iff \forall w\in s: M,w\models\phi$.
\end{definition}

\begin{proposition}[cf.\ Prop.\ 2.7 in \cite{Ciardelli:25neighborhood}] Every declarative $\aa\in\Lad$ is truth-conditional. Conversely, every truth-conditional formula of $\La$ is equivalent to a declarative.
\end{proposition}

\noindent
Not all formulas of  \inqnla\ are truth-conditional. As a simple example, consider the formula $?p$ (which abbreviates $p\lori\neg p$). A simple calculation yields the following support conditions:
$$M,s\models{?p}\iff p\text{ has the same truth value in all worlds in }s$$
Intuitively, ${?p}$ is supported at $s$ if it is \emph{settled} in $s$ whether $p$ is true or false.
Clearly, ${?p}$ is supported at every singleton state, and so true at each world, in every model; yet, it is not supported by every state.  

Obviously, formulas that are not truth-conditional cannot be equivalent to a declarative. However, any formula $\phi$ is equivalent to a finite inquisitive disjunction of declaratives, called the \emph{resolutions} of $\phi$.

\begin{definition}[Resolutions] The set $\R(\phi)$ of \emph{resolutions} of a formula $\phi\in\La$ is defined as follows: 
\begin{itemize}
\item $\R(\aa)=\{\aa\}$ if $\aa$ is an atom, $\bot$, or a modal formula $(\phi\Toa\psi)$;
\item $\R(\phi\land\psi)=\{\aa\land\bb\mid\aa\in\R(\phi),\bb\in\R(\psi)\}$;
\item $\R(\phi\lori\psi)=\R(\phi)\cup\R(\psi)$;
\item $\R(\phi\to\psi)=\{\bigwedge_{\aa\in\R(\phi)}(\aa\to f(\aa))\mid f:\R(\phi)\to\R(\psi)\}$.
\end{itemize}
\end{definition}

\noindent
Note that $\R(\phi)$ is always a finite non-empty set of declaratives, and for any declarative $\aa$, $\R(\aa)=\{\aa\}$.

\begin{proposition}[Normal form, cf.\ Prop.\ 2.9 in \cite{Ciardelli:25neighborhood}] \label{prop:normal form} For any $\phi\in\La$, $\phi\equiv\Lori\R(\phi)$.
\end{proposition}

\noindent
In terms of resolutions we may also define a further modal operator, $\Box_a$, which will play a role below:
$$\Box_a\phi\::=\;\bigvee\{\ibox_a\aa\mid\aa\in\R(\phi)\}$$
This operator has been considered in previous work on inquisitive modal logic \cite{CiardelliRoelofsen:15idel,Ciardelli:14aiml,Ciardelli:25}.\footnote{In this previous work, $\Box$ is taken as primitive, while here we take it as a defined operator. Either choice has advantages. The advantage of our present choice is that having fewer primitives simplifies the completeness proof below.}
In particular, it has been argued to give a natural generalization of the standard knowledge modality of epistemic logic,  allowing for a uniform analysis of knowledge ascriptions involving declarative complements (\emph{knowing that}) and interrogative complements (\emph{knowing whether/who/what}, etc).

Note that $\Box_a\phi$ is a declarative. A simple calculation yields the following truth conditions:
\begin{itemize}
\item $M,w\models\Box_a\phi\iff M,\bigcup\Sigma_a(w)\models\phi$
\end{itemize}

\noindent
To conclude, we mention a fact that will be useful later on: given that we start with a finite set of atoms, there are only finitely many formulas of a given modal depth, up to logical equivalence. 

\begin{proposition}[cf.\ Corollary 5.17 in \cite{Ciardelli:25neighborhood}]\label{prop:finiteness} The quotient of $\Lan$ under equivalence is finite.
\end{proposition}

%%%%%%%%%%%%%%%%%%%%%%%%%%%%%%%%%%%%%%
%%%%%%%%%%%%%%%%%%%%%%%%%%%%%%%%%%%%%%
%%%%%%%%%%%%%%%%%%%%%%%%%%%%%%%%%%%%%%

\section{Inquisitive action logic}

In this section we turn to the logic which is the focus of this paper: \emph{inquisitive action logic}, \inqal. This logic is a special case of the multi-agent inquisitive neighborhood logic \inqnla\ introduced in the previous section, obtained by focusing on a particular class of neighborhood models: those whose neighborhood functions represent actual effectivity functions of some concurrent game structure. For simplicity, we focus on the multi-agent case, i.e.,  we assume from now on that $\A$ contains at least two agents; for some comments on the trivial, but somewhat special, single-agent case, see Footnote \ref{fn:single agent axiomatization}. 

\begin{definition} A \emph{concurrent game model} (or \emph{cg-model}, for short) is a pair $\M=(\S,V)$ consisting of a concurrent game structure $\S=(W,\act,\out)$ and a valuation function $V:\P\to\wp(W)$. 
\end{definition}

\noindent
A concurrent game model $\M=(\S,V)$ with $\S=(W,\act,\out)$ determines a corresponding main-model $M_\M=(W,(\Sigma_a^\S)_{a\in\A},V)$, whose neighborhood functions are the actual effectivity functions of agents. %Clearly, not every main-model is induced from some cg-model, but only those whose neighborhood functions satisfy the conditions specified in Theorem \ref{theor:characterization}. 
We can thus interpret the formulas of $\La$ relative to a cg-model, via its associated main-model.

\begin{definition} Let $\M$ be a cg-model, $s$ a state in $\M$, and $\phi\in\La$. We write $\M,s\models\phi$ if $M_\M,s\models\phi$.
\end{definition}

\noindent Entailment and equivalence in \inqal, denoted $\models_\inqal$ and $\equiv_\inqal$ respectively, are defined in the same way as for \inqnla, but with respect to cg-models instead of main-models. Note that, being determined by a smaller class of models, the logic \inqal\ is an extension of \inqnla.

In the setting of concurrent game models, our modal formulas take on a specific significance. 
%in connection with the actions available to the agents and their ability to force or allow certain outcomes, and to determine certain aspects of the outcome. 
We will illustrate this by means of some examples (for the sake of space, we omit the simple calculations involved).
First, consider a modal formula of the form $\idia_a\aa$, where $\aa\in\Lad$. We have:
$$\M,w\models\idia_a\aa\iff\exists \tau_a\in\act(a,w):O_a(w,\tau_a)\subseteq|\aa|_\M
$$
%\begin{eqnarray*}
%\M,w\models\idia_a\aa&\iff& \exists s\in\Sigma_a^\S(w):s\models\aa\\
%&\iff& \exists s\in\Sigma_a^\S(w):s\subseteq|\aa|_\M\\
%&\iff& \exists \tau_a\in\act(a,w):O_a(w,\tau_a)\subseteq|\aa|_\M
%\end{eqnarray*}
\noindent
Thus, $\idia_a\aa$ expresses the fact that agent $a$ has an action available that guarantees the truth of $\aa$ (i.e., if that action is executed, the outcome will make $\aa$ true, regardless what the other agents do). In short, $\idia_a\aa$ says that $a$ can force $\aa$, recovering the standard semantics for ability familiar from \cite{Brown:88} and \cite{Pauly:02}.\footnote{In fact, it is possible to show that the $\idia$-fragment of our language is equi-expressive with the fragment of coalition logic that contains only modalities for individual agents. See \Sec 3 of \cite{Ciardelli:25neighborhood} for an analogous result in the case of \inqnl.}

As a second example, consider the modal formula $(\aa\Toa\bb)$, where $\aa,\bb\in\Lad$. We have:
$$\M,w\models(\aa\Toa\bb)\iff \forall\tau_a\in\act(a,w): O_a(w,\tau_a)\subseteq|\aa|_\M\text{ implies  }O_a(w,\tau_a)\subseteq|\beta|_\M$$
What this formula expresses is that, in order to force $\aa$, agent $a$ necessarily has to force $\beta$ as well. By negating such statements, we can express the fact that the agent can force certain outcomes without precluding others. 
Consider, e.g., the formula $\neg(\aa\Toa\neg\bb)$. We have:
$$\M,w\models\neg(\aa\Toa\neg\bb)\iff \exists\tau_a\in\act(a,w): O_a(w,\tau_a)\subseteq|\aa|_\M\text{ and }O_a(w,\tau_a)\cap|\beta|_\M\neq\emptyset$$
The formula expresses the fact that agent $a$ can perform an action which forces $\aa$ without precluding~$\beta$. More generally, a formula of the form $\neg(\aa\Toa\neg\bb_1\lori\dots\lori\neg\bb_n)$ expresses the fact that it is possible for $a$ to force $\aa$ without precluding any of $\bb_1,\dots,\bb_n$. In this way, the central properties that the Socially Friendly Coalition Logic of Goranko and Enqvist \cite{GorankoEnqvist:18} is designed to capture can be expressed in \inqal.

Let us now consider formulas involving the modalities $\ibox_a$ and $\Box_a$. First consider the case in which the argument is a declarative $\aa\in\Lad$. In this case, $\R(\aa)=\{\aa\}$, and so the two modalities coincide: by definition of $\Box_a$, we have $\Box_a\aa=\ibox_a\aa$. Recall from Section 2 that, in the case of actual effectivity functions, the union of the neighborhoods simply coincides with the set $O(w)$ of all outcomes possible at $w$: for any agent $a$, $\bigcup\Sigma_a^\S(w)=O(w)$. Using this fact, we obtain the following:
$$\M,w\models\Box_a\aa\iff\M,w\models\ibox_a\aa\iff O(w)\subseteq|\aa|_\M$$
Thus, regardless of the agent $a$, the (identical) formulas $\ibox_a\aa$ and $\Box_a\aa$ express the fact that $\aa$ is \emph{unavoidable}: it will be true at the next stage regardless what the agents do. 

When the argument is not a declarative, however, the modalities $\Box_a$ and $\ibox_a$ come apart. Let us illustrate this with the case of a polar question $?\aa$, where $\aa\in\Lad$. Since $\R(?\aa)=\{\aa,\neg\aa\}$, the formula $\Box_a{?\aa}$ amounts to $\ibox_a\aa\lor\ibox_a\neg\aa$, which says that either $\aa$ is unavoidably  true at the next stage, or it is unavoidably false. More formally:
$$\M,w\models \Box_a{?\aa}\iff\text{$\aa$ has the same truth value in all worlds in }O(w)$$
Thus, $\Box_a{?\aa}$ says that whether $\aa$ will be true at the next stage is \emph{predetermined}, i.e., settled \emph{a priori} regardless of the actions of the agents.
By contrast, for the modal formula $\ibox_a{?\aa}$ we have the following:
$$\M,w\models\ibox_a{?\aa}\iff \forall \tau_a\in\act(a,w): \text{$\aa$ has the same truth value in all worlds in }O_a(w,\tau_a)$$
Thus, what $\ibox_a{?\aa}$ expresses is that once we fix the action of $a$, the truth value of $\aa$ at the next stage is settled: the actions of other agents cannot affect it. These findings generalize: if $\phi$ denotes a question, then $\Box_a\phi$ expresses the fact that $\phi$ is settled \emph{a priori}, regardless of what the agents do, while $\ibox_a\phi$ expresses the fact that $\phi$ is settled once we fix the action of agent $a$. 

We can now see how the analysis of agentive determination we suggested in the introduction can be formalized in \inqal. Our guiding idea was this: at a certain stage in a process, an agent $a$ determines a question $\phi$ if (i) the answer to $\phi$ is not settled \emph{a priori} but (ii) it becomes settled once we fix the action of agent $a$. In our logic, (i) is captured by $\neg\Box_a\phi$, and (ii) by $\ibox_a\phi$. We can thus define a modality $\wbox_a$ capturing agentive determination in the following way:
$$\wbox_a\phi\;:=\;\neg\Box_a\phi\land \ibox_a\phi$$
For instance, the claim that agent $a$ determines whether $\aa$ will be true at the next stage is expressed by  $\wbox_a{?\aa}:=\neg\Box_a{?\aa}\land\ibox_a{?\aa}$, which says that the truth value of $\aa$ is not predetermined (i.e., not the same in all possible outcomes), but it is determined once we fix the action of $a$ (i.e., it is constant within each neighborhood for $a$). %, but that there are actions of $a$ yielding different truth values. 
This is, in my view, a very natural analysis of the determination claim.\footnote{Interestingly, exactly the same modal pattern, $\neg\Box_a\phi\land\ibox_a\phi$ is argued in inquisitive epistemic logic to capture the idea that an agent \emph{wonders about}, or is \emph{interested in}, a certain question \cite{CiardelliRoelofsen:15idel}. It is striking that such seemingly different notions across different modal domains plausibly share the same logical structure. Revealing this common structure is one of the payoffs of the formal analysis of question-oriented modal notions made possible by inquisitive modal logic.}

To illustrate this analysis further, consider the scenario from the introduction, where Alice ($a$) and Bob ($b$) are creating clay animals together. Let's say that there are three shapes that Alice is able to mould---cat, dog, and cow---and three colors available to Bob---red, blue, and green. A simple modeling of the scenario is one where Alice has three actions available to her: $\tau_{\text{cat}}$ (mould a cat), $\tau_{\text{dog}}$ (mould a dog), and $\tau_{\text{cow}}$ (mould a cow); similarly, Bob has three actions: $\tau_{\text{red}}$ (paint red), $\tau_{\text{blue}}$ (paint blue), $\tau_{\text{green}}$ (paint green). Each joint action by Alice and Bob, for instance $\tau_{\text{cat}}\tau_{\text{red}}$, results in a specific outcome, for instance $\out(w,\tau_{\text{cat}}\tau_{\text{red}})=\text{red cat}$. In total, there are $3\times 3=9$ possible outcomes (red cat, blue dog, etc.) which make up the total set $O(w)$. For each agent, fixing an action leaves us with only 3 possible outcomes; for instance, $O_a(w,\tau_{cat})=\{\text{red cat},\text{blue cat},\text{green cat}\}$, while $O_b(w,\tau_{red})=\{\text{red cat},\text{red dog},\text{red cow}\}$. The neighborhoods corresponding to these outcome sets for Alice and Bob are depicted in Figure \ref{fig:1}.

\begin{figure}[t]
\centering
\begin{tikzpicture}[>=latex,scale=.9]

 % --- Common rectangle size ---
 % width = 1.4, height = 4.6
 % (same size used everywhere)

 % Shape rectangles (vertical columns, identical size)
 \draw[opaque,rounded corners] (-2.2,2.3) rectangle (-0.8,-2.3);
 \draw[opaque,rounded corners] (-0.7,2.3) rectangle (0.7,-2.3);
 \draw[opaque,rounded corners] (0.8,2.3) rectangle (2.2,-2.3);

 % Nodes (3x3 grid)
 % Top row (red)
 \node at (-1.5,1.5) {\footnotesize \parbox{.6cm}{red\\[-.1cm]dog}};
 \node at (0,1.5)    {\footnotesize \parbox{.6cm}{red\\[-.1cm]cat}};
 \node at (1.5,1.5)  {\footnotesize \parbox{.6cm}{red\\[-.1cm]cow}};

 % Middle row (blue)
 \node at (-1.5,0) {\footnotesize \parbox{.6cm}{blue\\[-.1cm]dog}};
 \node at (0,0)    {\footnotesize \parbox{.6cm}{blue\\[-.1cm]cat}};
 \node at (1.5,0)  {\footnotesize \parbox{.6cm}{blue\\[-.1cm]cow}};

 % Bottom row (green)
 \node at (-1.5,-1.5) {\footnotesize \parbox{.6cm}{green\\[-.1cm]dog}};
 \node at (0,-1.5)    {\footnotesize \parbox{.6cm}{green\\[-.1cm]cat}};
 \node at (1.5,-1.5)  {\footnotesize \parbox{.6cm}{green\\[-.1cm]cow}};

 \node at (0,-3) {$\Sigma_a^\S(w)$};
\end{tikzpicture}
\hspace{3cm}
\begin{tikzpicture}[>=latex,scale=.9]

 % Color rectangles (horizontal rows, EXACT same size as above)
 \draw[opaque,rounded corners] (-2.3,2.2) rectangle (2.3,0.8);
 \draw[opaque,rounded corners] (-2.3,0.7) rectangle (2.3,-0.7);
 \draw[opaque,rounded corners] (-2.3,-0.8) rectangle (2.3,-2.2);

 % Same nodes
 \node at (-1.5,1.5) {\footnotesize \parbox{.6cm}{red\\[-.1cm]dog}};
 \node at (0,1.5)    {\footnotesize \parbox{.6cm}{red\\[-.1cm]cat}};
 \node at (1.5,1.5)  {\footnotesize \parbox{.6cm}{red\\[-.1cm]cow}};

 \node at (-1.5,0) {\footnotesize \parbox{.6cm}{blue\\[-.1cm]dog}};
 \node at (0,0)    {\footnotesize \parbox{.6cm}{blue\\[-.1cm]cat}};
 \node at (1.5,0)  {\footnotesize \parbox{.6cm}{blue\\[-.1cm]cow}};

 \node at (-1.5,-1.5) {\footnotesize \parbox{.6cm}{green\\[-.1cm]dog}};
 \node at (0,-1.5)    {\footnotesize \parbox{.6cm}{green\\[-.1cm]cat}};
 \node at (1.5,-1.5)  {\footnotesize \parbox{.6cm}{green\\[-.1cm]cow}};

 \node at (0,-3) {$\Sigma_b^\S(w)$};
\end{tikzpicture}
\caption{\label{fig:1}The neighborhoods (i.e., outcome sets) associated to Alice's actions (left) and Bob's action (right) in our toy example.}
\end{figure}
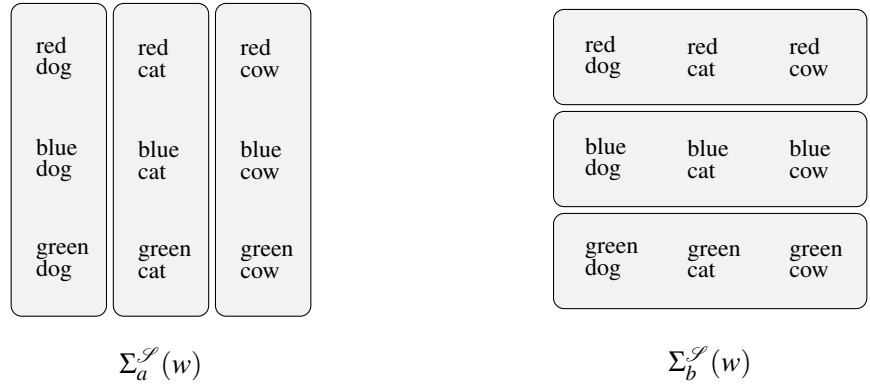

Now we can consider a language with six propositional atoms, \textsf{cat}, \textsf{dog}, \textsf{cow}, \textsf{red}, \textsf{blue}, \textsf{green}, with the obvious interpretation (e.g., \textsf{cat} expresses the proposition \quotes{the animal is a cat}). Using inquisitive disjunction, we can build two formulas \textsf{shape} and \textsf{color} expressing, respectively, the questions \quotes{what the shape is} and \quotes{what the color is}: 
$$\textsf{shape}:=(\textsf{cat}\lori\textsf{dog}\lori\textsf{cow})\qquad\textsf{color}:=(\textsf{red}\lori\textsf{blue}\lori\textsf{green})$$
An information state $s$ in our model supports the question \textsf{shape} just in case in all worlds in $s$, the animal has the same shape; similarly, $s$ supports \textsf{color} if in all worlds in $s$, the animal has the same color.

By embedding these questions under modalities we can now describe what aspects of the outcome Alice and Bob respectively determine through their actions: Alice determines the shape ($\wbox_a\textsf{shape}$) but not the color ($\neg\wbox_a\textsf{color}$), while Bob determines the color ($\wbox_b\textsf{color}$) but not the shape ($\neg\wbox_b\textsf{shape}$).

In order to see that $\wbox_a\textsf{shape}$ (i.e., $\neg\Box_a\textsf{shape}\land\ibox_a\textsf{shape}$) is true at our world $w$, we reason as follows. First, since the shape of the animal is not constant across all possible outcomes, we have $\M,O(w)\not\models\textsf{shape}$, which ensures $\M,w\models\neg\Box_a\textsf{shape}$; this captures the fact that in our scenario, the shape is not predetermined, but depends on the actions of the agents. Second, within each of the three outcome sets $O_a(w,\tau_{cat}),O_a(w,\tau_{dog}),O_a(w,\tau_{cow})$, corresponding to the three actions for $a$, the shape of the animal \emph{is} constant; therefore, each of these sets supports \textsf{shape}, ensuring $\M,w\models\ibox_a\textsf{shape}$; this captures the fact that once we fix Alice's action, the shape of the outcome is settled, regardless of what Bob does. 

By contrast, $\wbox_a\textsf{color}$ (i.e., $\neg\Box_a\textsf{color}\land\ibox_a\textsf{color}$) is false at $w$, since its second conjunct is false. Consider any outcome set for $a$, for instance $O_a(w,\tau_{cat})$. The color of the animal is not constant across this set, and so this set does not support $\textsf{color}$. Since not all the outcome sets for $a$ (in fact, none of them) support $\textsf{color}$, $\M,w\not\models\ibox_a\textsf{color}$. This reflects the fact even if we fix Alice's action, the color of the outcome is not settled---it depends in part (and in fact, in our case, completely) on what Bob does.

%%%%%%%%%%%%%%%%%%%%%%%%%%%%%%%%%%
%%%%%%%%%%%%%%%%%%%%%%%%%%%%%%%%%%

\section{Axiomatization and decidability}

We start by presenting an axiomatization of \inqnla, which is the basis for our axiomatization of \inqal.

The propositional axioms include instances of axioms for intuitionistic propositional logic, with $\lori$ in the role of intuitionistic disjunction, and all instances of the following, where $\aa\in\Lad$ and $\phi,\psi\in\La$:
\begin{itemize}
\item $\neg\neg\aa\to\aa$\hfill (Declarative double negation)
\item $(\aa\to\phi\lori\psi)\to(\aa\to\phi)\lori(\aa\to\psi)$\hfill (Split)
\end{itemize}
%
%\begin{description}
%\item[(\texttt{DDN})] $\neg\neg\aa\to\aa$
%\item[(\texttt{Split})] $(\aa\to\phi\lori\psi)\to(\aa\to\phi)\lori(\aa\to\psi)$
%\end{description}
The modal axioms include all instances of the following schemata, capturing the behavior of the modalities $\Toa$ as strict conditionals (see \cite{Ciardelli:25neighborhood} for discussion of these axioms in the uni-modal case and \cite{LitakVisser:18} for a more general study of strict conditionals on an intuitionistic basis):
\begin{itemize}
\item $(\phi\Toa\psi)\land(\psi\Toa\chi)\to(\phi\Toa\chi)$\hfill (Transitivity)
\item $(\phi\Toa\psi)\land(\phi\Toa\chi)\to(\phi\Toa(\psi\land\chi))$\hfill (Right conjunction)
\item $(\phi\Toa\chi)\land(\psi\Toa\chi)\to((\phi\lori\psi)\Toa\chi)$\hfill (Left disjunction)
\end{itemize}
The inference rules are \emph{modus ponens} ($\phi,\phi\to\psi/\psi$) and \emph{conditional necessitation} ($\phi\to\psi/\phi\Toa\psi$).

If $\psi\in\La$ is derivable in this system we write $\vdashn\psi$. For $\Phi,\Psi\subseteq\La$, we write $\Phi\vdashn\Psi$ in case for some finite $\Phi_0\subseteq\Phi$ and $\Psi_0\subseteq\Psi$ we have $\vdashn\bigwedge\Phi_0\to\Lori\Psi_0$. Instead of $\{\phi_1,\dots,\phi_n\}\vdashn\{\psi_1,\dots,\psi_m\}$ we write simply $\phi_1,\dots,\phi_n\vdashn\psi_1,\dots,\psi_m$. We write $\phi\preqn\psi$ if we have both $\phi\vdashn\psi$ and $\psi\vdashn\phi$.

The following soundness and strong completeness theorem generalizes the one for \inqnl. The proof is the obvious multi-agent adaptation of the one for \inqnl\ given in \cite{Ciardelli:25neighborhood}.

\begin{theorem}\label{theor:completeness-n} For all $\Phi\subseteq\La$ and $\psi\in\La$, $\Phi\models_{\inqnla}\psi\iff\Phi\vdashn\psi$.
\end{theorem}

\noindent
To obtain a system for \inqal, we add three axiom schemata reflecting the properties of actual effectivity functions we identified in Theorem~\ref{theor:characterization}. For all declaratives $\aa$, all formulas $\phi_1,\dots,\phi_n$, all agents $a,b$, and all pairwise distinct agents $a_1,\dots,a_{n+1}$, the following are axioms:\footnote{\label{fn:single agent axiomatization}Recall that we assume that $\A$ contains at least two agents. If $\A$ contains a single agent $a$, an axiomatization of \inqal\ is obtained by following two axioms: (i) $\idia_a\top$ (existence of actions); (ii) $\ibox_a?{\aa}$ for any declarative $\aa$ (determinacy). The former captures the condition that $\Sigma_a(w)\neq\emptyset$, the latter the condition that each $s\in\Sigma_a(w)$ is a singleton. As mentioned in Footnote \ref{fn:single agent characterization}, these conditions guarantee that $\Sigma_a$ is the actual effectivity function of some CGS. The completeness proof for this case follows the one below, except that the proof of the Lemma \ref{lemma:key} needs to be adapted. We leave this as an exercise.}
\begin{itemize}
\item $\idia_a\top$\hfill (Existence of actions)
\item $\ibox_a\aa\leftrightarrow\ibox_b\aa\;$\footnote{Note that the restriction to declaratives in this axiom is crucial; for instance, $\ibox_a{?p}\leftrightarrow\ibox_b{?p}$ is not valid in \inqal. Alternatively, we may take as axioms all formulas of the form $\Box_a\phi\leftrightarrow\Box_b\phi$; in this case, no restriction is needed.} \hfill(Uniform range)
\item $\idia_{a_1}\phi_1\land\dots\land\idia_{a_n}\phi_n\,\to\,\neg{\idia_{a_{n+1}}}\neg(\phi_1\land\dots\land\phi_n)$ 
%for all formulas $\phi_1,\dots,\phi_n\in\La$ and distinct agents $a_1,\dots,a_n\in\A$
\hfill (Independence)
\end{itemize}
We write $\vdasha\psi$ if $\psi$ is derivable in this extended system. Similarly, the notations $\Phi\vdasha\Psi$ and $\phi\preqa\psi$ are defined as for $\vdashn$, but with reference to the extended axiom system.
%if we have $\vdasha(\phi\leftrightarrow\psi)$. 
Our task is now to prove that this extended system is sound and (weakly) complete for \inqal.\footnote{An interesting question that we will leave open is whether Theorem \ref{theor:completeness} can be extended to a strong completeness result.}

\begin{theorem}[Soundness and completeness for \inqal] \label{theor:completeness} For all $\phi\in\La$, $\models_{\inqal}\phi\iff\vdasha\phi$.
\end{theorem}

\noindent For soundness, we just have to check that the axioms are valid and the inference rules preserve validity. We leave this as an exercise, noting only that the three extra axioms for \inqal\ owe their validity to the three corresponding properties of actual effectivity functions listed in Theorem \ref{theor:characterization}. Towards completeness, we adapt a finite canonical model construction from \cite{Ciardelli:25neighborhood} (in turn building on ideas from \cite{Ciardelli:14aiml}). 

\begin{definition} For $n\in\mathbb N$, an \emph{$n$-bounded complete theory of declaratives} (or $n$CTD for short), is a set of declaratives $\Gamma\subseteq\Land$ satisfying three conditions: (i) deductive closure relative to $\Land$: if $\Gamma\vdasha\aa$ and $\aa\in\Land$, then $\aa\in\Gamma$; (ii) consistency: $\bot\not\in\Gamma$; and (iii) completeness: for all $\aa\in\Land$, $\aa\in\Gamma$ or $\neg\aa\in\Gamma$. We denote the set of all $n$CTDs by $\K_n$. 
%\begin{itemize}
%\item deductive closure relative to $\Land$: if $\Gamma\vdasha\aa$ and $\aa\in\Land$, $\aa\in\Gamma$;
%\item consistency: $\bot\not\in\Gamma$;
%\item completeness relative to $\Land$: for all $\aa\in\Land$, $\aa\in\Gamma$ or $\neg\aa\in\Gamma$.
%\end{itemize}
\end{definition}

\noindent We now show that each $\K_n$ is a non-empty finite set, and that the sets $\K_n$ and $\K_m$ are disjoint for $n\neq m$. 

%For each $n$, the set $\K_n$ is non-empty. Indeed, the following Lindenbaum-type lemma ensures that any proof-theoretically consistent subset of $\Land$ can be extended to am $n$CTD. 

\begin{lemma}\label{lemma:Lindenbaum} If $\Delta\subseteq\Land$ and $\Delta\not\vdasha\bot$, then $\Delta\subseteq\Gamma$ for some $\Gamma\in\K_n$. In particular, $\K_n\neq\emptyset$.
\end{lemma}

\begin{proof} A simple adaptation of the usual saturation argument.
\end{proof}

%\noindent
%Further, each set $\K_n$ is finite, as we are now going to show.

\begin{lemma}\label{lemma:finiteness} Each set $\K_n$ is finite. 
\end{lemma}
\begin{proof} Since $\vdasha$ extends $\vdashn$, the equivalence relation $\preqn$ refines $\preqa$. By Theorem \ref{theor:completeness-n}, $\preqn$ coincides with the semantic equivalence relation $\equiv_{\inqnla}$. By Prop.\ \ref{prop:finiteness}, the quotient of $\Lan/\equiv_{\inqnla}$ is finite, and a fortiori, so is the quotient $\Land/\preqa$. Since an $n$CTD is a subset of $\Land$ which is deductively closed, it is fully determined by the equivalence classes of its elements modulo $\preqa$, i.e., it is fully determined by a subset of the quotient $\Land/\preqa$. Since these subsets are finitely many, so are the $n$CTDs.
\end{proof}

\begin{lemma}\label{lemma:uniqueness} $\K_n\cap\K_m=\emptyset$ for $n\neq m$.
\end{lemma}

\begin{proof} Let $a$ be an arbitrary agent and let $\ibox_a^h\top$ abbreviate $\ibox_a\dots\ibox_a\top$ with $h$ occurrences of $\ibox_a$. If $\Gamma\in\K_n$, by deductive closure relative to $\Land$ we have $\ibox_a^h\top\in\Gamma\iff h\le n$, and so $n=\max\{h\mid \ibox_a^h\top\in\Gamma\}$. As a consequence, if $\Gamma\in \K_n\cap \K_m$ then $n=\max\{h\mid \ibox_a^h\top\in\Gamma\}=m$.
\end{proof}

\noindent
We now define for each $n\in\mathbb{N}$ a canonical model suitable for formulas in $\Lan$. % based on $m$CDTs for $m\le n$.

\begin{definition}
For $n\in\mathbb{N}$, we define the main-model $M_n=(W_n,(\Sigma_{a}^n)_{a\in\A},V_n)$ as follows:  $W_n=\bigcup_{m\le n}\K_m$; $V_n(p)=\{\Gamma\in W_n\mid p\in\Gamma\}$; finally, the neighborhood map $\Sigma_{a}^n$ is defined as follows:

\begin{itemize}
\item if $\Gamma\in\K_0$ then $\Sigma_{a}^n(\Gamma)=\{\K_0\}$
\item if $\Gamma\in\K_{m+1}$ then 
$\Sigma_{a}^n(\Gamma)=\{S\subseteq\K_m, S\neq\emptyset\mid \text{for all }(\psi\To_a\!\chi)\in\Gamma:%\\
%\phantom{\Sigma_{a}^n(\Gamma)=\{S\subseteq\K_h, S\neq\emptyset\mid} \bigcap S\vdash\psi\text{ implies }\bigcap S\vdash\chi\}
\bigcap S\vdasha\psi\text{ implies }\bigcap S\vdasha\chi\}$
\end{itemize}
\end{definition}

\noindent
For any $n\in\mathbb{N}$, the model $M_n$ is finite by Lemma \ref{lemma:finiteness}. The following four lemmas play a key role in the completeness proof. The proofs are simple adaptations of those of the corresponding results for \inqnl\ (Lemmas 6.10, 6.12, 6.15, and 6.16 in \cite{Ciardelli:25neighborhood}), We provide the details in Appendix A for completeness.

\begin{lemma}[Intersection Lemma]\label{lemma:intersection}~\\ 
For a set $\Delta\subseteq\Land$, define $S_\Delta^n=\{\Gamma\in\K_n\mid\Delta\subseteq\Gamma\}$. For every $\phi\in\Lan$ we have $\Delta\vdasha\phi\iff \bigcap S_\Delta^n\vdasha\phi$.\footnote{To be fully precise, for $S\subseteq\K_n$, we define $\bigcap S=\{\aa\in\Land\mid \aa\in\Gamma\text{ for all }\Gamma\in S\}$. This definition implies that $\bigcap\emptyset=\Land$, ensuring that the lemma holds also when we have $\Delta\vdasha\bot$, and therefore $S_\Delta^n=\emptyset$. In this case, the sets of formulas $\Delta$ and $\bigcap S_\Delta^n=\bigcap \emptyset=\Land$ indeed derive the same formulas from $\Lan$: all of them.} 
In particular, since $S_\emptyset^n=\K_n$, for every $\phi\in\Lan$ we have $\vdasha\phi\iff\bigcap\K_n\vdasha\phi$.
\end{lemma}

\begin{lemma}[Existence Lemma]\label{lemma:existence}~\\ For $m\le n$, if $\Gamma\in\K_m$ and $\neg(\phi\Toa\psi)\in\Gamma$, then there is $S\in\Sigma_a^n(\Gamma)$ such that $\bigcap S\vdasha\phi$ and $\bigcap S\not\vdasha\psi$. 
\end{lemma}

\begin{lemma}[Range Lemma]\label{lemma:range}~\\ For all $\Gamma\in\K_{m}$ with $m>0$ we have 
$\bigcup\Sigma_a^n(\Gamma)=\{\Gamma'\in\K_{m-1}\mid\forall\aa\in\Lad:\ibox_a\aa\in\Gamma\text{ implies }\aa\in\Gamma'\}$.
\end{lemma}

\begin{lemma}[Support Lemma]\label{lemma:support}~\\ For all $m\le n$, all non-empty states $S\subseteq\K_m$, and all formulas $\phi\in\L_\A^m$: $M_n,S\models\phi\iff\bigcap S\vdasha\phi$.
\end{lemma}

\noindent
The part of the proof which is genuinely novel, and crucial for our purposes, lies in showing that the canonical models $M_n$ so constructed are induced by some corresponding concurrent game models.

%satisfy conditions (a)-(c) of Theorem \ref{theor:characterization}, and so, they are induced by corresponding concurrent game models.

\begin{lemma}[Representation Lemma] \label{lemma:key} For each $n\in\mathbb{N}$ there is a CGM $\M_n$ such that $M_{\M_n}=M_n$.
\end{lemma}

\begin{proof}
It suffices to show that the maps $(\Sigma_a^n)_{a\in\A}$ of our canonical model satisfy the conditions (a)-(c) of Theorem \ref{theor:characterization}. Consider a world $\Gamma\in W_n$. We have $\Gamma\in\K_m$ for some $m\le n$. If $m=0$ then $\Sigma_a^n(\Gamma)=\{\K_0\}$ for every agent $a$ and conditions (a)-(c) are obviously satisfied. So, we may assume $m>0$.

\medskip\noindent
(a) \emph{Existence of actions.} Since $\Gamma$ contains the axiom $\idia_a\top$, which is short for $\neg(\top\Toa\bot)$, the existence of some $S\in\Sigma_a^n(\Gamma)$ follows directly from Lemma \ref{lemma:existence}.

\medskip\noindent
(b) \emph{Uniform range.} Take any agents $a,b\in\A$. For any declarative $\aa$, $\ibox_a\aa\leftrightarrow\ibox_b\aa$ is an axiom, so $\ibox_a\aa\preqa\ibox_b\aa$. Since the formulas $\ibox_a\aa$ and $\ibox_b\aa$ also have the same modal depth, $\Gamma$ contains one iff it contains the other. Now using Lemma \ref{lemma:range} for both agents $a$ and $b$ we have:
\begin{eqnarray*}
\bigcup\Sigma_a^n(\Gamma)&= & \{\Gamma'\in\K_{m-1}\mid\forall\aa\in\Lad:\ibox_a\aa\in\Gamma\text{ implies }\aa\in\Gamma'\}\\
&= & \{\Gamma'\in\K_{m-1}\mid\forall\aa\in\Lad:\ibox_b\aa\in\Gamma\text{ implies }\aa\in\Gamma'\}\quad =\quad \bigcup\Sigma_b^n(\Gamma)
\end{eqnarray*}

\noindent
(c) \emph{Independence.} Take any neighborhoods $S_1\in\Sigma_{a_1}^n(\Gamma), \dots, S_k\in\Sigma_{a_k}^n(\Gamma)$. We must prove $S_1\cap\dots\cap S_k\neq\emptyset$. We start by proving the following claim:
$$\qquad\qquad\bigcap S_1 \cup \dots \cup \bigcap S_k\not\vdasha\bot\qquad\quad(*)$$
Towards a contradiction, suppose this is false. Then there are formulas $\aa_1\in\bigcap S_1,\dots,\aa_k\in\bigcap S_k$ such that $\aa_1,\dots,\aa_k\vdasha\bot$ (it suffices to consider a single $\aa_i$ from each $\bigcap S_i$, since $\bigcap S_i$ is closed under conjunction).

By definition of $\Sigma_{a_i}^n$, each $S_i$ is non-empty and thus $\bigcap S_i\not\vdasha\bot$. 
Since $S_i\in\Sigma_{a_i}^n(\Gamma)$ and $\bigcap S_i\vdasha\aa_i$ but $\bigcap S_i\not\vdasha\bot$, again by definition of $\Sigma_{a_i}^n$ it follows that $(\aa_i\To_{a_i}\bot)\not\in\Gamma$. 
Also, since $\Gamma\in\K_m$, by definition of $\Sigma_{a_i}^n(\Gamma)$ we have $S_i\subseteq\K_{m-1}$, which means that $\text{md}(\aa_i)\le m-1$, and thus $\text{md}(\aa_i\To_{a_i}\bot)\le m$. 
Since $\Gamma\in\K_m$ and $(\aa_i\To_{a_i}\bot)\in\L_\A^{!m}$, from $(\aa_i\To_{a_i}\bot)\not\in\Gamma$ we can conclude by the completeness condition on $\Gamma$ that $\neg(\aa_i\To_{a_i}\bot)\in\Gamma$, that is, %by the completeness condition on CTDs. we have 
$\idia_{a_i}\aa_i\in\Gamma$.
So, for each $i\le k$ we have $\idia_{a_i}\aa_i\in\Gamma$. 
But now, by deductive closure relative to $\L_{\A}^{!m}$, $\Gamma$ contains both the axiom
$$\idia_1\aa_1\land\dots\land\idia_{k-1}\aa_{k-1}\to\neg\idia_k\neg(\aa_1\land\dots\land\aa_{k-1})$$
and its antecedent. So, $\Gamma$ must contain the consequent, $\neg\idia_k\neg(\aa_1\land\dots\land\aa_{k-1})$. 

On the other hand, from  $\aa_1,\dots,\aa_k\vdasha\bot$ we have that $\aa_k\vdasha\neg(\aa_1\land\dots\land\aa_{k-1})$, and since $\aa_k\in\bigcap S_k$, also  $\neg(\aa_1\land\dots\land\aa_{k-1})\in\bigcap S_k$. Reasoning again as in the previous paragraph, from this and the fact that $S_k\in\Sigma_{a_k}^n(\Gamma)$ we may conclude that $\Gamma$ contains $\idia_k\neg(\aa_1\land\dots\land\aa_{k-1})$.

So, $\Gamma$ contains both $\idia_k\neg(\aa_1\land\dots\land\aa_{k-1})$ and the negation of this formula; by deductive closure, it must contain $\bot$. But this is a contradiction, since $\Gamma$ is consistent by assumption. Thus, $(*)$ is true.

Thus, $\bigcap S_1 \cup \dots \cup \bigcap S_k$ is a consistent subset of $\L_\A^{!m-1}$. By Lemma \ref{lemma:Lindenbaum} there is a $\Delta\in\K_{m-1}$ with $\bigcap S_1\cup\dots\cup\bigcap S_k\subseteq \Delta$. To conclude, we show that $\Delta\in S_1\cap \dots\cap S_k$, which implies that $S_1\cap \dots\cap S_k\neq\emptyset$.

Take any $i\le k$. We know that $\bigcap S_i\subseteq\Delta$ and we want to prove $\Delta\in S_i$. Towards a contradiction, suppose that $\Delta\not\in S_i$. By Lemma \ref{lemma:finiteness}, we know that $S_i$ is a finite set, 
$S_i=\{\Theta_1,\dots,\Theta_\ell\}$ for some $\ell$ and $\Theta_1,\dots,\Theta_\ell\in\K_{m-1}$. For $j\le\ell$, $\Theta_j\neq\Delta$ and therefore we can find $\theta_j\in\L_{\A}^{!m-1}$ with $\theta_j\in\Theta_j$ and $\neg\theta_j\in\Delta$. Now let $\theta=\theta_1\lor\dots\lor\theta_\ell$. We have $\theta\in\bigcap S_i$ but $\theta\not\in\Delta$, which is a contradiction since $\bigcap S_i\subseteq\Delta$.
\end{proof}

%To establish this, we show that the three specific axioms of \inqal\ ensure that the canonical model $M_n$ satisfies conditions (a)-(c) of our Theorem \ref{theor:characterization}. The details are spelled out in Appendix \ref{app:characterization}. 

\noindent
Finally, we can now put everything together and establish completeness.

\begin{proof}[Proof of Theorem \ref{theor:completeness}, left-to-right.] Suppose $\not\vdasha\phi$ and let $n=\md(\phi)$. By Lemma \ref{lemma:intersection}, we have that $\bigcap \K_n\not\vdasha\phi$, which by Lemma \ref{lemma:support} implies $M_n,\K_n\not\models\phi$. By Lemma \ref{lemma:key}, there is a CGM $\M_n$ that induces $M_n$, and so we have $\M_n,\K_n\not\models\phi$. Since $\phi$ can be falsified in some CGM, $\not\models_{\inqal}\phi$.\end{proof}

% and suppose $\R(\phi)=\{\aa_1,\dots,\aa_h\}$; note that by definition of resolutions, $\md(\aa_i)\le n$ for each $i\le h$. By Lemma \ref{lemma:provable nf}, for $i\le h$ we have $\not\vdasha\aa_i$; by the declarative double negation axiom, this implies $\not\vdasha\neg\neg\aa_i$, which  amounts to $\neg\aa_i\not\vdasha\bot$. Since $\neg\aa_i\not\vdasha\bot$ and $\neg\aa_i\in\Land$, by Lemma \ref{lemma:Lindenbaum} we can find a $\Gamma_i\in\K_n$ with $\neg\aa_i\in\Gamma_i$. Now let $S=\{\Gamma_1,\dots,\Gamma_h\}$. For $i\le h$ we have $\bigcap S\not\vdasha\aa_i$. (If we had $\bigcap S\vdasha\aa_i$, then since $\bigcap S\subseteq\Gamma_i$, also $\Gamma_i\vdasha\aa_i$; but since $\Gamma_i\ni\neg\aa_i$, it would follow that $\Gamma_i\vdasha\bot$, whence $\bot\in\Gamma_i$ by deductive closure, contradicting the consistency of $\Gamma_i$.) Hence, by Lemma \ref{lemma:support}, we have $M_n,S\not\models\aa_i$ for $1\le i\le h$. By Normal Form (Prop.\ \ref{prop:normal form}), it follows that $M_n,S\not\models\phi$. By Lemma \ref{lemma:key}, there is a CGM $\M_n$ that induces $M_n$, and so we have $\M_n,S\not\models\phi$. Since $\phi$ can be falsified in some state of some CGM, we have $\not\models_{\inqal}\phi$.\end{proof}

\noindent
Note that the proof yields, for any $\phi$ which is not provable in \vdasha, a \emph{finite} countermodel. %the set $\P$ of atoms occurring in $\phi$ is finite, therefore the main-model $M_n$ is finite, and so is the cg-model $\M_n$, which is based on the same set of worlds. 
Therefore, our proof also establishes the finite model property of \inqal. In turn, the finite model property together with our (recursive) axiomatization yields the decidability of \inqal. 

\begin{corollary}[Finite model property] If $\not\models_{\inqal}\phi$, $\phi$ can be refuted within a finite cg-model.
\end{corollary}

%\noindent
%Together, our (recursive) axiomatization of \inqal\ and the finite model property yield decidability. 

\begin{corollary}[Decidability] The problem of deciding if a given $\phi\in\La$ is valid in $\inqal$ is decidable. 
\end{corollary}

%stated by the following analogue of the standard truth-lemma.

%%%%%%%%%%%%%%%%%%%%%%%%%%%%%%%%%%%%
%%%%%%%%%%%%%%%%%%%%%%%%%%%%%%%%%%%%

%%%%%%%%%%%%%%%%%%%%%%%%%%%%%%%%%%%%
%%%%%%%%%%%%%%%%%%%%%%%%%%%%%%%%%%%%

\section{Connection with Socially Friendly Coalition Logic}

In this section we relate \inqal\ to Socially Friendly Coalition Logic (\textsf{SFCL}) \cite{GorankoEnqvist:18}, a generalization of coalition logic \cite{Pauly:02} developed by Goranko and Enqvist based on instantial neighborhood logic \cite{Benthem:17}. Modal formulas in \textsf{SFCL} have the form $[C](\sigma;\pi_1,\dots,\pi_n)$ where $C\subseteq\A$ and $\sigma,\pi_1,\dots,\pi_n$ are formulas. The primitive connectives are $\neg$ and $\lor$. To compare it with \inqal, we restrict to the individual-agent fragment of \textsf{SFCL}, where $C$ is a singleton $\{a\}$, which we may identify with the agent $a$.\footnote{The connection discussed in this section extends straighforwardly to one between full \textsf{SFCL} and the natural generalization of \inqal\ with coalitions. The translations would work in the same way, and in particular, the same exponential blowup discussed below would result when translating a formula from the coalitional version of \inqal\ to \textsf{SFCL}.}
Formulas are interpreted relative to CGMs in a standard truth-conditional fashion. The clause for modal formulas is as follows:
$$\M,w\Vdash[a](\sigma;\pi_1,\dots,\pi_n)\iff \exists\tau_a\in\act(a,w):O_a(w,\tau_a)\subseteq|\sigma|_\M\text{ and }O_a(w,\tau_a)\cap|\pi_i|_\M\neq\emptyset\text{ for }i\le n$$
Based on the observations in \Sec 4, we can define a translation $(\cdot)^*$ from the individual-agent fragment of \textsf{SFCL} to \inqal\ in the following way: $p^*=p$; $(\neg\phi)^*=\neg\phi^*$; $(\phi\lor\psi)^*=\phi^*\lor\psi^*$; and, finally:
$$[a](\sigma;\pi_1,\dots,\pi_n)^*=\neg(\sigma^*\Toa \neg\pi_1^*\lori\dots\lori\neg\pi_n^*)$$
It is easy to check that the translation $\sigma^*$ of a \textsf{SFCL}-formula is always a declarative with the same truth conditions as $\sigma$. Translating in the opposite direction, from \inqal\ to \textsf{SFCL}, is more tricky. First, since the language of \textsf{SFCL} includes only statements, we can only expect to faithfully translate declaratives, and not arbitrary formulas. Moreover, even though $\phi\Toa\psi$ is a declarative, the formulas $\phi$ and $\psi$ need not be declarative, and so, they need not have a translation. Still, we can cook up a translation via their resolutions. We define a translation from the declarative fragment of \inqal\ to \textsf{SFCL} as follows: $p^\star=p$; $\bot^\star=(p\land\neg p)$ for an arbitrary $p\in\P$; $(\aa\land\bb)^\star=\aa^\star\land\bb^\star$; $(\aa\to\bb)^\star=\neg(\aa^\star\land\neg\bb^\star)$; and, crucially,
$$(\phi\Toa\psi)^\star=\bigwedge_{i=1}^n\neg[a](\aa_i^\star;\neg\bb_1^\star,\dots,\neg\bb_m^\star)$$
where $\{\aa_1,\dots,\aa_n\}=\R(\phi)$ and $\{\bb_1,\dots,\bb_m\}=\R(\psi)$. One can prove that for any declarative $\aa\in\Lad$, $\aa$ and its translation $\aa^*$ have the same truth conditions. The proof is identical to the one given for the translation of \inqnl\ into instantial neighborhood logic (Proposition 10.2 in \cite{Ciardelli:25neighborhood}).

In sum, \inqal\ and \textsf{SFCL} are equi-expressive with respect to statements about individual agents. Still, the way things are expressed in these logics is rather different. Note, in particular, that the number of resolutions $\R(\phi)$ of a formula $\phi$ can grow exponentially relative to the size of $\phi$, and thus, so can the length of the translation of a modal formula containing $\phi$. It is plausible to conjecture that this blowup is unavoidable, and so, that \inqal\ is in general exponentially more succinct than \textsf{SFCL}. It is also worth noting that there is currently no complete axiomatization of \textsf{SFCL} (an axiomatization is presented in \cite{GorankoEnqvist:18}, but by the authors' admission (p.c.) the completeness proof contains a mistake). %So, there is currently no complete axiomatization of \textsf{SFCL}. 
It can be hoped that our Theorem \ref{theor:characterization} can be used to establish completeness at least for the individual-agent fragment of \textsf{SFCL}.

%%%%%%%%%%%%%%%%%%%%%%%%%%%%%%%%%%%%
%%%%%%%%%%%%%%%%%%%%%%%%%%%%%%%%%%%%

\section{Conclusions and future work}
We have motivated and investigated \inqal, a logic that allows us to reason not only about what agents can \emph{force} through their actions, but also about what they \emph{enable} or \emph{determine}. An obvious goal for future work is to generalize \inqal\ from individual agents to coalitions. Semantically, this is straightforward. The difficulty lies in establishing a characterization theorem analogous to Theorem \ref{theor:characterization}, which played a crucial role in our proof of completeness and decidability. A further important goal is to extend \inqal\ with temporal operators, yielding an inquisitive version of strategic multi-agent logics like \textsf{ATL} \cite{Alur:02,GorankoDrimmelen:06}. In a different direction, it would be interesting to ask if \inqal, or an extension, can regiment claims to the effect that an agent can partly influence (though perhaps not fully determine) the answer to a question.

\paragraph{Acknowledgements.} Funding from the European Research Council (ERC) under the Horizon Europe research and innovation programme (Project InqML, Grant Agreement No.\ 101116774) is gratefully acknowledged. Thanks to Valentin Goranko for detailed discussions on this topic, and to four anonymous reviewers for precious comments and suggestions.

%%%%%%%%%%%%%%%%%%%%%%%%%%%%%%%%%%%%
%%%%%%%%%%%%%%%%%%%%%%%%%%%%%%%%%%%%
%%%%%%%%%%%%%%%%%%%%%%%%%%%%%%%%%%%%
%%%%%%%%%%%%%%%%%%%%%%%%%%%%%%%%%%%%

\appendix

\section{Appendix: Proofs of standard results}

\subsection{Preliminaries}

We start with some basic results about derivability and resolutions. Essentially the same results appear in many completeness results for inquisitive propositional logic \cite{Ciardelli:23book} and inquisitive modal logics \cite{Ciardelli:14aiml,Ciardelli:18aiml,Ciardelli:25neighborhood}. In each case, we outline the proof and provide a reference where an analogous proof is spelled out.

\begin{lemma}[Provable normal form] \label{lemma:provable nf} For all $\phi\in\La$, $\phi\preqa\Lori\R(\phi)$.
\end{lemma}

\begin{proof} This uses only the propositional component of the axiomatization. See Lemma 4.3.8 in \cite{Ciardelli:23book}.
\end{proof}

\begin{lemma}\label{lemma:proto split} For all $\phi\in\La$, if $\vdasha\phi$ then $\vdasha\aa$ for some $\aa\in\R(\phi)$.
\end{lemma}

\begin{proof} It suffices to check that this property holds for axioms and is preserved by inference rules. When an axiom is a declarative $\aa$, the claim is trivially true since $\R(\aa)=\{\aa\}$. Since all our modal axioms are declaratives, the only axioms we need to check are the propositional ones. This is a simple and standard exercise (see the proof of Lemma 5.6 in \cite{Ciardelli:18aiml}). 
As for the inference rules, if $\phi$ was obtained by Conditional Necessitation then $\phi$ is declarative and the claim is trivially true. If $\phi$ was obtained by Modus Ponens from $\psi$ and $\psi\to\phi$, then by induction hypothesis some resolutions $\aa\in\R(\psi)$ and $\bb\in\R(\psi\to\phi)$ are derivable. By definition of resolutions of an implication, $\bb$ is a conjunction with one conjunct of the form $\aa\to\gamma$ where $\gamma\in\R(\phi)$. Since $\aa$ and $\aa\to\gamma$ are derivable, $\gamma$ is derivable by Modus Ponens.
\end{proof}

\begin{lemma}\label{lemma:provable split} If $\Gamma\subseteq\Lad$ and $\Gamma\vdasha\phi$, then $\Gamma\vdasha\aa$ for some $\aa\in\R(\phi)$
\end{lemma}

\begin{proof} If $\Gamma\vdasha\phi$, this means that there is a finite subset $\Gamma_0\subseteq\Gamma$ such that $\bigwedge\Gamma_0\to\phi$ is derivable. Let $\gamma=\bigwedge\Gamma_0$. Since $\Gamma$ is a set of declaratives, $\gamma$ is a declarative, so $\R(\gamma)=\{\gamma\}$. Then, by definition of resolutions $\R(\gamma\to\phi)=\{\gamma\to\aa\mid\aa\in\R(\phi)\}$. By the previous lemma, since $\gamma\to\phi$ is derivable, some particular resolution $\gamma\to\aa$ is derivable. This implies that $\Gamma\vdasha\aa$, and since $\aa\in\R(\phi)$ we are done.
\end{proof}

\begin{definition}[Resolutions for sets] A \emph{resolution function} for a set $\Phi\subseteq\La$ is a function $f$ assigning to each $\phi\in\Phi$ a corresponding resolution $\aa\in\R(\phi)$. The image of $\Phi$ under a resolution function is called a \emph{resolution} of $\Phi$. More formally, the set of resolutions of $\phi$ is defined in the following way:
$\R(\Phi)=\{f[\Phi]\mid f\text{ a resolution function for }\Phi\}$. Note that if $\Gamma\in\Phi$, then for each $\phi\in\Phi$ there is some resolution $\aa\in\R(\phi)$ with $\aa\in\Gamma$.
\end{definition}

\begin{lemma}\label{lemma:provable specification} If $\Phi,\Psi\subseteq\La$ and $\Phi\not\vdasha\Psi$, then $\Gamma\not\vdasha\Psi$ for some $\Gamma\in\R(\Phi)$.
\end{lemma}

\begin{proof} Repeated application of Lemma \ref{lemma:provable nf}, by standard properties of $\lori$. See Lemma 4.3.7 in \cite{Ciardelli:23book}. 
\end{proof}

\subsection{Intersection Lemma (Lemma \ref{lemma:intersection})}

Consider a set of declaratives $\Delta\subseteq\Land$ and its corresponding set of $n$-bounded complete extensions, $S_\Delta^n=\{\Gamma\in\K_n\mid\Delta\subseteq\Gamma\}$. We want to show that for all $\phi\in\Lan$:
$\Delta\vdasha\phi\iff\bigcap S_\Delta^n\vdasha\phi$.

The direction $\Rightarrow$ is obvious since $\Delta\subseteq\bigcap S_\Delta^n$. For the converse, suppose for a contradiction that for some $\phi$ we had $\bigcap S_\Delta^n\vdasha\phi$ but $\Delta\not\vdasha\phi$. Since $\bigcap S_\Delta^n$ is a set of declaratives, by Lemma \ref{lemma:provable split} we have $\bigcap S_\Delta^n\vdasha \aa$ for some $\aa\in\R(\phi)$. Since $\Delta\not\vdasha\phi$, it follows by Lemma \ref{lemma:provable nf} that $\Delta\not\vdasha\aa$. By the axiom $\neg\neg\aa\to\aa$, this implies $\Delta\not\vdasha\neg\neg\aa$, and therefore $\Delta\cup\{\neg\aa\}\not\vdasha\bot$. Since $\phi\in\Lan$ we have $\aa\in\Land$ and thus also $\Delta\cup\{\neg\aa\}\subseteq\Land$. 
So by Lemma \ref{lemma:Lindenbaum} there is $\Gamma\in\K_n$ such that $\Delta\cup\{\neg\aa\}\subseteq\Gamma$. Now since $\Gamma\in S_\Delta^n$ and $\bigcap S_\Delta^n\vdasha\aa$ we have $\Gamma\vdasha\aa$. So we have $\Gamma\vdasha\neg\aa$ and $\Gamma\vdasha\aa$, whence $\Gamma\vdasha\bot$ and, by deductive closure relative to $\Land$, $\bot\in\Gamma$. But this is a contradiction since $\Gamma\in\K_n$ is consistent by assumption.\hfill$\Box$

%%%%%%%%%%%%%%%%%%%%%%%%%%%%%%%%%%%%
%%%%%%%%%%%%%%%%%%%%%%%%%%%%%%%%%%%%

\subsection{Existence Lemma (Lemma \ref{lemma:existence})}

\noindent
We follow closely the proof of Lemma 6.12 in \cite{Ciardelli:25neighborhood}, adapting it to our multi-agent setting and to our modal depth-bounded canonical model construction. We first establish some preliminary results. 

Let $\Gamma$ be an $n$CTD, with $n>0$. Given two sets $\Phi,\Psi\subseteq\L_\A^{n-1}$, we write ${\Phi\Togamma\Psi}$ if there are finite subsets $\Phi_0\subseteq\Phi$ and $\Psi_0\subseteq\Psi$ such that the formula $\bigwedge\Phi_0\Toa\Lori\Psi_0$
%$(\phi_1\land\dots\land\phi_h)\Toa(\psi_1\lori\dots\lori\psi_m)$
 is in $\Gamma$. Note the following fact.

\begin{lemma}\label{lemma:vdash togamma} If $\Phi,\Psi\subseteq\L_\A^{n-1}$ and $\Phi\vdasha\Psi$, then $\Phi\Togamma\Psi$.
\end{lemma}
 
\begin{proof}  If $\Phi\vdasha\Psi$, there are finite subsets $\Phi_0\subseteq\Phi$ and $\Psi_0\subseteq\Psi$ such that $\vdasha\bigwedge\Phi_0\to\Lori\Psi_0$. By Conditional Necessitation, also $\vdasha\bigwedge\Phi_0\Toa\Lori\Psi_0$. Since $\Phi,\Psi\subseteq\L_\A^{n-1}$, the modal depth of the formula $\Phi_0\Toa\Lori\Psi_0$ is at most $n$, and so by deductive closure, this formula must be in $\Gamma$, witnessing $\Phi\Togamma\Psi$.
 \end{proof}
\noindent
Importantly, the relation $\Togamma$ also enjoys the following  cut-like property.

\begin{lemma}\label{lemma:cut} For any two sets $\Phi,\Psi\subseteq\L_\A^{n-1}$ and formula $\chi\in\L_\A^{n-1}$: 
$\Phi\cup\{\chi\}\Togamma\Psi$ and $\Phi\Togamma\Psi\cup\{\chi\}$ implies $\Phi\Togamma\Psi$.
\end{lemma}

\begin{proof}
Suppose $\Phi\cup\{\chi\}\Togamma\Psi$ and $\Phi\Togamma\Psi\cup\{\chi\}$. This means that there are finite sets of formulas $\Phi_0,\Phi_1\subseteq\Phi$ and $\Psi_0,\Psi_1\subseteq\Psi$ such that $\Gamma$ contains the following formulas:
$$(\chi\land\bigwedge\Phi_0)\Toa\Lori\Psi_0\qquad\qquad \bigwedge\Phi_1\Toa(\chi\lori\Lori\Psi_1)$$
We prove that $\Gamma$ must contain $\bigwedge(\Phi_0\cup\Phi_1)\Toa \Lori(\Psi_0\cup\Psi_1)$, thus witnessing ${\Phi\Togamma\Psi}$. To ease notation, we spell out the details for the case in which the relevant sets are all singletons $\Phi_0=\{\phi_0\},\Phi_1=\{\phi_1\},\Psi_0=\{\psi_0\},\Psi_1=\{\psi_1\}$, but the general case is analogous. 

So, we know that $\Gamma$ contains the formulas $(\phi_1\land\chi\Toa\psi_1)$ and $(\phi_2\Toa\psi_2\lori\chi)$, and we want to show that it contains $(\phi_1\land\phi_2\Toa\psi_1\lori\psi_2)$. Since this formula is a declaratives with modal depth $\le n$, and since $\Gamma$ is closed under deduction relative to such formulas, it suffices to show that:
$$(\phi_1\land\chi\Toa\psi_1),(\phi_2\Toa\psi_2\lori\chi)\vdasha(\phi_1\land\phi_2\Toa\psi_1\lori\psi_2)$$
%First note that the following formula is provable in the propositional component of the proof system using the standard axioms for $\land$ and $\lori$:
%\begin{equation}\label{distribution}
%\phi_1\land(\psi_2\lori\chi)\to\psi_2\lori(\phi_1\land\chi)
%\end{equation}
A derivation is given in Figure \ref{fig:derivation}. In the derivation, we indicate explicitly only the modal axioms and rules involved in the reasoning, omitting reference to propositional axioms. We write $(\texttt{MP})$ to indicate \emph{modus ponens} and $(\texttt{CN})$ for \emph{conditional necessitation}. For simplicity, we use the formulas $(\phi_1\land\chi\Toa\psi_1)$ and $(\phi_2\Toa\psi_2\lori\chi)$ as if they were premises; this is legitimate since we will not use the conditional necessitation rule (\texttt{CN}) on these formulas or anything inferred from them. Rewriting the argument with the relevant formulas used throughout as conditional antecedents is tedious but straightforward.\end{proof}

\begin{figure}[t]
\begin{mdframed}
\begin{tabular}{lll}
1. & $\phi_1\land\chi\,\Toa\,\psi_1$ & (premise)\\
2. & $\phi_2\,\Toa\,\psi_2\lori\chi$ & (premise)\\
3. & $\phi_1\land\phi_2\Toa\phi_1$ & (\texttt{CN}) from axiom $\phi_1\land\phi_2\to\phi_1$\\
4. & $\phi_1\land\phi_2\Toa\phi_2$ & (\texttt{CN})  from axiom $\phi_1\land\phi_2\to\phi_2$\\
5. & $\phi_1\land\phi_2\Toa\psi_2\lori\chi$ & (\texttt{MP}) from Transitivity, 4, 2\\
6. & $\phi_1\land\phi_2\Toa\phi_1\land(\psi_2\lori\chi)$ & (\texttt{MP}) from Right Conjunction, 3, 5\\
7. & $\phi_1\land(\psi_2\lori\chi)\Toa\psi_2\lori (\phi_1\land\chi)$ & (\texttt{CN}) from prop.\ validity $\phi_1\land(\psi_2\lori\chi)\to\psi_2\lori(\phi_1\land\chi)$\\
8. & $\phi_1\land\phi_2\Toa\psi_2\lori (\phi_1\land\chi)$ & (\texttt{MP}) from Transitivity, 6, 7\\
9. & $\psi_1\Toa\psi_1\lori\psi_2$ & (\texttt{CN}) from axiom $\psi_1\to\psi_1\lori\psi_2$\\
10. & $\psi_2\Toa\psi_1\lori\psi_2$ & (\texttt{CN}) from axiom $\psi_2\to\psi_1\lori\psi_2$\\
11. & $\phi_1\land\chi\Toa\psi_1\lori\psi_2$ & (\texttt{MP}) from Transitivity, 1, 9\\
12. & $\psi_2\lori (\phi_1\land\chi)\Toa \psi_1\lori\psi_2$ & (\texttt{MP}) from Left Disjunction, 10, 11\\
13. & $\phi_1\land\phi_2\,\Toa\,\psi_1\lori\psi_2$ & (\texttt{MP}) from Transitivity, 8, 12
\end{tabular}
\caption{\label{fig:derivation}Sketch of a proof showing $(\phi_1\land\chi\Toa\psi_1),(\phi_2\Toa\psi_2\lori\chi)\vdasha(\phi_1\land\phi_2\Toa\psi_1\lori\psi_2)$}
\end{mdframed}
\end{figure}

\begin{lemma}[Splitting lemma]\label{lemma:splitting} Let $\Gamma$ be a $n$CTD with $n>0$ and take $\Phi,\Psi\subseteq\L_\A^{n-1}$ with $\Phi\not\Togamma\Psi$. The set $\L_\A^{n-1}$ can be partitioned into sets $\Left$ and $\Right$ such that $\Phi\subseteq\Left$, $\Psi\subseteq\Right$, and $\Left\not\Togamma\Right$.
\end{lemma}

\vspace{-.5cm}
\begin{proof}
Fix an enumeration $(\chi_i)_{i\in\mathbb N}$ of $\L_\A^{n-1}$. % (we assume $\L_\To$ is countable, though this is not strictly necessary). 
Define a sequence of sets $(\Left_i)_{i\in\mathbb{N}}$ and $(\Right_i)_{i\in\mathbb{N}}$ as follows:
\begin{itemize}
\item $\Left_0=\Phi, \Right_0=\Psi$
\item if $\Left_i\cup\{\chi_i\}\not\Togamma\Right_i$ we let $\Left_{i+1}:=\Left_i\cup\{\chi_i\}$ and $\Right_{i+1}=\Right_i$
\item if $\Left_i\cup\{\chi_i\}\Togamma\Right_i$ we let $\Left_{i+1}:=\Left_i$ and $\Right_{i+1}=\Right_i\cup\{\chi_i\}$
\end{itemize}
We show by induction on $i$ that $\Left_i\not\Togamma\Right_i$. For $n=0$ this is true by assumption. Now suppose this is true for $i$ and consider $i+1$. If $\Left_i\cup\{\chi_i\}\not\Togamma\Right_i$, the claim is obvious by definition of $\Left_{i+1}$ and $\Right_{i+1}$. So, suppose $\Left_i\cup\{\chi_i\}\Togamma\Right_i$. Since by induction hypothesis $\Left_i\not\Togamma\Right_i$, Lemma \ref{lemma:cut} implies $\Left_i\not\Togamma\Right_i\cup\{\chi_i\}$, which by definition amounts to $\Left_{i+1}\not\Togamma\Right_{i+1}$.

Now let $\Left=\bigcup_{i\in\mathbb{N}}\Left_i$ and $\Right=\bigcup_{i\in\mathbb{N}}\Right_i$. By construction, $\Phi\subseteq\Left$ and $\Psi\subseteq\Right$. We have $\Left\not\Togamma\Right$, otherwise there would be an $i\in\mathbb{N}$ such that $\Left_i\Togamma\Right_i$, contrary to what we just saw. Moreover, $\Left$ and $\Right$ form a partition of $\L_\A^{n-1}$. By construction, every formula of $\L_\A^{n-1}$ occurs in either set. Moreover, no formula cannot occur in both: to see why, suppose for a contradiction that for some $\chi\in\L_\A^{n-1}$ we have $\chi\in\Left\cap\Right$; since $\chi\Toa\chi$ is a valid declarative with modal depth $\le n$, and since $\Gamma$ is deductively closed with respect to such formulas, we would need to have $(\chi\Toa\chi)\in\Gamma$, contradicting $\Left\Togamma\Right$.
\end{proof}

\noindent
With these preliminaries at hand, we are now ready to complete the proof of the existence lemma.

\medskip\noindent\emph{Proof of Lemma \ref{lemma:existence}.}
Let $\Gamma$ be an $n$CTD with $\neg(\phi\Toa\psi)\in\Gamma$. This implies $n>0$ (otherwise $\Gamma$ would contain only propositional formulas) and furthermore $\phi$ and $\psi$ must be in $\L_\A^{n-1}$. Since $\Gamma$ is consistent, $(\phi\Toa\psi)\not\in\Gamma$, and so $\{\phi\}\not\Togamma\{\psi\}$.
Now extend $\{\phi\}$ and $\{\psi\}$ to sets $\Left$ and $\Right$ as in the previous lemma.
 
Since $\Left\not\Togamma\Right$, by Lemma \ref{lemma:vdash togamma} we have $\Left\not\vdasha\Right$. 
By Lemma \ref{lemma:provable specification} there is a set $\Delta\in\R(\Left)$ with $\Delta\not\vdasha\Right$. Since $\Left\subseteq\L_\A^{n-1}$, we have $\Delta\subseteq\L_\A^{!n-1}$. 
We can now take $S=S_\Delta^{n-1}=\{\Gamma'\in\K_{n-1}\mid \Delta\subseteq\Gamma'\}$. 
We need to verify that (i) $\bigcap S\vdasha\phi$, (ii) $\bigcap S\not\vdasha\psi$ and (iii) $S\in\Sigma_a^n(\Gamma)$.

\begin{itemize}
\item For (i), we have $\phi\in\Left$. Since $\Delta\in\R(\Left)$, for some $\aa\in\R(\phi)$ we have $\aa\in\Delta$. By Lemma \ref{lemma:provable nf}, $\Delta\vdasha\phi$, and thus, since $\phi\in\L_\A^{n-1}$, by 
Lemma \ref{lemma:intersection} also $\bigcap S\vdasha\phi$.

\item For (ii), we have $\psi\in\Right$. Since $\Delta\not\vdasha\Right$, also $\Delta\not\vdasha\psi$. Since $\psi\in\L_\A^{n-1}$, Lemma \ref{lemma:intersection} gives $\bigcap S\not\vdasha\psi$.

\item For (iii), first note that since $\Delta\not\vdasha\Right$, we have $\Delta\not\vdasha\bot$, so by Lemma \ref{lemma:Lindenbaum}, $S\neq\emptyset$. Next, suppose $(\chi\Toa\xi)\in\Gamma$ and $\bigcap S\vdasha\chi$. We need to show that $\bigcap S\vdasha\xi$. Since $(\chi\Toa\xi)\in\Gamma$ and $\Gamma\in\K_n$ we have $\chi,\xi\in\L_\A^{n-1}$.
Since $\bigcap S\vdasha\chi$ and $\chi\in\L_\A^{n-1}$, Lemma \ref{lemma:intersection} gives $\Delta\vdasha\chi$. Since by construction $\Delta\not\vdasha\Right$, it follows that $\chi\not\in\Right$, and since $\Right$ and $\Left$ partition the set $\L_\A^{n-1}$ we have $\chi\in\Left$. Now we must have $\xi\in\Left$ as well, for if we had $\xi\in\Right$ it would follow from $(\chi\Toa\xi)\in\Gamma$ that $\Left\Togamma\Right$, contrary to what we know. Since $\xi\in\Left$ and $\Delta\in\R(\Left)$, for some $\aa\in\R(\xi)$ we have $\aa\in\Delta$, so by Lemma \ref{lemma:provable nf}, $\Delta\vdasha\xi$. Finally, since $\xi\in\L_\A^{n-1}$, Lemma \ref{lemma:intersection} implies $\bigcap S\vdasha\xi$, as desired.\hfill$\Box$
\end{itemize}

%%%%%%%%%%%%%%%%%%%%%%%%%%%%%%%%%%%%
%%%%%%%%%%%%%%%%%%%%%%%%%%%%%%%%%%%%

\subsection{Range Lemma (Lemma \ref{lemma:range})}

We follow the proof of Lemma 6.15 in \cite{Ciardelli:25neighborhood}, adapting it to our multi-agent and depth-bounded setting. 

Consider a $\Gamma\in\K_{m}$ with $m>0$. We want to show the identity: 
$$\bigcup\Sigma_a^n(\Gamma)=\{\Gamma'\in\K_{m-1}\mid\forall\aa\in\Lad:\ibox_a\aa\in\Gamma\text{ implies }\aa\in\Gamma'\}$$

\begin{proof} $(\subseteq)$ Suppose $\Gamma'\in\bigcup\Sigma_a^n(\Gamma)$, that is, $\Gamma'\in S$ for some $S\in\Sigma_a^n(\Gamma)$. By definition of $\Sigma_a^n$, since $\Gamma\in\K_m$ we have $\Gamma'\in\K_{m-1}$.
Now let $\aa$ be a declarative and suppose $\ibox_a\aa\in\Gamma$, that is, $(\top\Toa\aa)\in\Gamma$. Since $S\in\Sigma_a^n(\Gamma)$ and $\bigcap S\vdasha\top$, it follows that $\bigcap S\vdasha\aa$. Since $\Gamma'\in S$, we have $\bigcap S\subseteq\Gamma'$, and so also $\Gamma'\vdasha\aa$. Since $\ibox_a\aa\in\Gamma$ and $\Gamma\in\K_m$, it follows that $\aa\in\L_\A^{!m-1}$. Since $\Gamma'\vdasha\aa$ and $\Gamma'$ is closed under deduction relative to $\L_\A^{!m-1}$, we have $\aa\in\Gamma'$.

$(\supseteq)$ Consider $\Gamma'\in\K_{m-1}$ and suppose for all $\aa\in\Lad$, $\ibox_a\aa\in\Gamma\text{ implies }\aa\in\Gamma'$. We must show that $\Gamma'\in S$ for some $S\in\Sigma_a^n(\Gamma)$.
First, we claim that $\emptyset\not\Togamma\{\neg\aa\mid\aa\in\Gamma'\}$. Towards a contradiction, suppose not: then there are $\aa_1,\dots,\aa_n\in\Gamma'$ such that $(\top\Toa\neg\aa_1\lori\dots\lori\neg\aa_n)\in\Gamma$. Since $\neg\aa_1\lori\dots\lori\neg\aa_n\vdasha\neg(\aa_1\land\dots\land\aa_n)$, by Conditional Necessitation and Transitivity we have $(\top\Toa\neg\aa_1\lori\dots\lori\neg\aa_n)\vdasha(\top\Toa\neg(\aa_1\land\dots\land\aa_n))$. Since $(\top\Toa\neg(\aa_1\land\dots\land\aa_n))\in\L_\A^{m}$ and $\Gamma$ is deductively closed with respect to $\L_\A^m$, we conclude $(\top\Toa\neg(\aa_1\land\dots\land\aa_n))\in\Gamma$, that is, $\ibox_a\neg(\aa_1\land\dots\land\aa_n)\in\Gamma$. By our assumption on $\Gamma'$, we must have $\neg(\aa_1\land\dots\land\aa_n)\in\Gamma'$. But this is impossible, since each $\aa_i$ is in $\Gamma'$ and $\Gamma'$ is consistent.

We have thus established the claim $\emptyset\not\Togamma\{\neg\aa\mid\aa\in\Gamma'\}$. By Lemma \ref{lemma:splitting}, we can partition the language $\L_\A^{m-1}$ into sets $\Left,\Right$ with $\Left\not\Togamma\Right$ and $\{\neg\aa\mid\aa\in\Gamma'\}\subseteq\Right$. Reasoning as in the previous lemma, we can find a 
$\Delta\in\R(\Left)$ with $\Delta\not\vdash\Right$, and we can show that the corresponding set of $m-1$-bounded complete extensions $S_\Delta^{m-1}$ is in $\Sigma_a^n(\Gamma)$. We now claim that $\Gamma'\in S_\Delta^{m-1}$. To show this, it suffices to show that $\Delta\cup\Gamma'\not\vdasha\bot$: if this holds, it follows by Lemma \ref{lemma:Lindenbaum} that there is a $\Gamma''\in\K_{m-1}$ with $\Delta\cup\Gamma'\subseteq\Gamma''$. Since two bounded elements of $\K_{m-1}$ cannot be properly included in one another, we must have $\Gamma'=\Gamma''$, and therefore $\Delta\subseteq\Gamma'$, showing that $\Gamma'\in S_\Delta^{m-1}$ as desired.

So, towards a contradiction, suppose $\Delta\cup\Gamma'\vdasha\bot$. Since $\Gamma'$ is closed under conjunction, this means that there is a single formula $\aa\in\Gamma'$ such that $\Delta\cup\{\aa\}\vdasha\bot$, and so, $\Delta\vdasha\neg\aa$. But this is impossible, since by construction $\neg\aa\in\Right$ and $\Delta\not\vdash\Right$. 

To conclude, we have found a state $S_\Delta^{m-1}$ such that $\Gamma'\in S_\Delta^{m-1}$ and $S_\Delta^{m-1}\in\Sigma_a^n(\Gamma)$, thus showing that $\Gamma'\in\bigcup \Sigma_a^n(\Gamma)$, as required.
\end{proof}

%%%%%%%%%%%%%%%%%%%%%%%%%%%%%%%%%%%%
%%%%%%%%%%%%%%%%%%%%%%%%%%%%%%%%%%%%

\subsection{Support Lemma (Lemma \ref{lemma:support})}

We must show that for all $m\le n$, all  formulas $\phi\in\L_\A^m$, and all non-empty states $S\subseteq\K_m$ we have 
$$M_n,S\models\phi\iff\bigcap S\vdasha\phi$$
The proof is by induction on $\phi$, simultaneously for all $S\subseteq\K_m$. The cases for atoms and connectives are standard (see the proof of Lemma 4.3.15 in \cite{Ciardelli:23book}). We spell out the inductive step for a modal formula $\phi=(\psi\Toa\chi)$. Note that since we are assuming $\phi\in\L_\A^m$ we have $m>0$ and $\psi,\chi\in\L_\A^{m-1}$. 

Suppose $\bigcap S\vdasha(\psi\Toa\chi)$. We must show $M_n,S\models(\psi\Toa\chi)$. For this, take an arbitrary world $\Gamma\in S$ and a state $T\in\Sigma_a^n(\Gamma)$ with $M_n,T\models\psi$. We need to show that $M_n,T\models\chi$.
By definition of $\Sigma_a^n$ we have $T\subseteq\K_{m-1}$. By induction hypothesis on $\psi$, from $M_n,T\models\psi$ we obtain $\bigcap T\vdasha\psi$. Since $\Gamma\in S$ we have $\bigcap S\subseteq\Gamma$, and since $\bigcap S\vdasha(\psi\Toa\chi)$ also $\Gamma\vdash(\psi\Toa\chi)$. Since $(\psi\Toa\chi)\in\L_\A^{!m}$ and $\Gamma\in\K_m$, it follows that $(\psi\Toa\chi)\in\Gamma$. By definition of $\Sigma_a^n$, from $T\in \Sigma_a^n(\Gamma)$, $(\psi\Toa\chi)\in\Gamma$, and $\bigcap T\vdasha\psi$ we can conclude $\bigcap T\vdasha\chi$. Finally, by induction hypothesis on $\chi$, this gives $M_n,T\models\chi$, as required. 

For the converse, suppose $\bigcap S\not\vdasha(\psi\Toa\chi)$. Then there is some $\Gamma\in S$ such that $(\psi\Toa\chi)\not\in\Gamma$. Since $(\psi\Toa\chi)\in\L_\A^{m}$ and $\Gamma\in\K_m$, by completeness we have $\neg(\psi\Toa\chi)\in\Gamma$. By the Existence Lemma (Lemma \ref{lemma:existence}) there is a state $T\in\Sigma_a^m(\Gamma)$ such that $\bigcap T\vdasha\psi$ and $\bigcap T\not\vdasha\chi$. By induction hypothesis on $\psi$ and $\chi$, this means that $M_n,T\models\psi$ and $M_n,T\not\models\chi$. Hence, $M_n,S\not\models(\psi\Toa\chi)$.\hfill$\Box$

%%%%%%%%%%%%%%%%%%%%%%%%%%%%%%%%%%%%
%%%%%%%%%%%%%%%%%%%%%%%%%%%%%%%%%%%%

%%%%%%%%%%%%%%%%%%%%%%%%%%%%%%%%%%%%
%%%%%%%%%%%%%%%%%%%%%%%%%%%%%%%%%%%%

\newpage
\nocite{*}
\bibliographystyle{eptcs}
\bibliography{aiml26-with-dois}

%\section{Proof of the Representation Lemma (Lemma \ref{lemma:key})}
%\label{app:characterization}
%

%$$\infer[(\texttt{MP})]{\psi}{\phi & \phi\to\psi}\qquad\qquad \infer[(\texttt{CN})]{\phi\To\psi}{\phi\to\psi}$$

%%%%%%%%%%%%%%%%%%%%%%%%%%%%%%%%%%%%%%%%%%
%%%%%%%%%%%%%%%%%%%%%%%%%%%%%%%%%%%%%%%%%%

%\bibliographystyle{eptcs}
%\bibliography{aiml26}

%%%%%%%%%%%%%%%%%%%%%%%%%%%%%%%%%%%%%%%%%%
%%%%%%%%%%%%%%%%%%%%%%%%%%%%%%%%%%%%%%%%%%
%%%%%%%%%%%%%%%%%%%%%%%%%%%%%%%%%%%%%%%%%%
%%%%%%%%%%%%%%%%%%%%%%%%%%%%%%%%%%%%%%%%%%

\end{document}